\documentclass[12pt]{amsart}
\usepackage{graphicx}

\usepackage{color}

\usepackage{graphicx}
\usepackage{epsfig}

\newtheorem{thm}{Theorem}[section]

\newtheorem{lemma}[thm]{Lemma}

\newtheorem{remark}[thm]{Remark}

\usepackage{amsmath}
\usepackage{amsxtra}
\usepackage{amscd}
\usepackage{amsthm}
\usepackage{amsfonts}
\usepackage{amssymb}
\usepackage{eucal}
\newcommand{\bmb}{\left( \begin{array}{rr}}
\newcommand{\enm}{\end{array}\right)}

\newcommand{\Z}{{\mathbb Z}}

\newcommand{\al}{{\alpha}}

\newcommand{\beq}{\begin{equation}}
\newcommand{\eeq}{\end{equation}}
\newcommand{\beqa}{\begin{eqnarray}}
\newcommand{\eeqa}{\end{eqnarray}}
\newcommand{\nnb}{\nonumber}
\newcommand{\pd}{\partial}

\numberwithin{equation}{section}

\begin{document}

\title{Arctic Curves in path models from the Tangent Method}
\author{Philippe Di Francesco} 
\address{Department of Mathematics, University of Illinois, Urbana, IL 61821, U.S.A.  \break  e-mail: philippe@illinois.edu
}
\author{Matthew F. Lapa}
\address{Department of Physics, University of Illinois , Urbana, IL 61821, U.S.A.  \break  e-mail: lapa2@illinois.edu}

\begin{abstract}

Recently, Colomo and Sportiello introduced a powerful method, known as the \emph{Tangent Method}, for computing the
arctic curve in statistical 
models which have a (non- or weakly-) intersecting lattice path formulation. We apply
the Tangent Method to compute arctic curves in various models: the domino tiling of the Aztec diamond for which we recover the 
celebrated arctic circle; a model of Dyck paths equivalent to the rhombus tiling of a half-hexagon for which we find an arctic half-ellipse;
another rhombus tiling model with an arctic parabola; the vertically symmetric alternating sign matrices, where we find the same arctic
curve as for unconstrained alternating sign matrices. The latter case involves lattice paths that are non-intersecting but that are allowed
to have osculating contact points, for which the Tangent Method was argued to still apply.
%
For each problem
we estimate the  large size asymptotics of a certain one-point function using LU decomposition 
of the corresponding Gessel-Viennot matrices, and a reformulation of the result amenable to asymptotic analysis.
\end{abstract}

\maketitle
\date{\today}
\tableofcontents


\section{Introduction}
It is now well known that under certain conditions tiling problems of finite plane domains 
display an ``Arctic curve" phenomenon \cite{CEP,JPS} when the size of the domains becomes very large,
namely a sharp separation between ``crystalline" (i.e. regularly tiled) phases, typically induced by corners of the domain,
and ``liquid" (i.e. disordered) phases, away from the boundary of the domain.
A particular subclass of such problems are the so-called ``dimer" models, where tiling configurations are replaced by a dual
notion of dimers, i.e. occupation of the edges of a given domain of a lattice by objects (dimers) in such a way that each 
vertex of the domain is covered by a unique such object. These have received considerable attention over the years,
culminating in general asymptotic results and a characterization of the arctic curves as solving some optimization problem
\cite{KO1,KO2,KOS}.

A common denominator between all these models is the existence of a reformulation in terms of 
Non-Intersecting Lattice Paths (NILP).
For tilings, it is due to the existence of conservation laws (local properties that propagate throughout the domain), giving rise
to configurations of paths (often called De Bruijn lines \cite{DeBruijn}) in bijection with the tilings. For dimers, it is related to the description of
configurations via so-called zig-zag paths \cite{Kdim}.

Recently, Colomo and Sportiello \cite{COSPO} came up with a novel approach to the determination of arctic curves, 
coined the ``Tangent Method". It is based on the non-intersecting lattice path formulation. The idea is to modify the lattice 
path configurations of a certain domain so as to impose that a single border path escapes and reaches a distant target point outside of the domain.
It is argued that for large size this path should leave the arctic curve at some point.
Away from the other paths, and for large size, the latter is most likely to follow a straight line between the point where it leaves the arctic curve 
and the distant target. This line is moreover argued to be tangent to the actual arctic curve. The Tangent method consists of determining for a fixed
target, the most likely boundary point of the domain at which the escaping path leaves the domain, thus obtaining a parametric family of lines
that are tangent to the arctic curve. The latter is recovered as the envelope of the family of lines obtained by moving the target, say along a line.
The intriguing feature of the method is that it seems to extend beyond ordinary non-intersecting lattice path models. In particular, in \cite{COSPO}
the authors argue that the method also applies to the case of so-called osculating paths such as those in bijection with configurations of the
6 Vertex model or Alternating Sign Matrices (ASMs) at the ice point \cite{zeilberger1,kuperberg1996another}. 
As a consequence they obtain an alternative derivation of the ASMs arctic
curve providing the same result as earlier calculations based on assumptions of a very different nature \cite{CP2010,CNP}.

In this paper, we address some concrete examples and apply the tangent method to determine arctic curves. We first treat the case of domino tilings 
of the Aztec diamond for which we recover the celebrated arctic circle result of \cite{CEP}. 
Then we go on to study two lattice path models, both attached to
rhombus tilings of particular domains of the triangular lattice. Finally we treat the case of Vertically Symmetric 
Alternating Sign Matrices (VSASMs), as an extension
of Colomo and Sportiello's calculation for Alternating Sign Matrices.
In doing so, we have to deal with two main complications: 
\begin{itemize}
\item As the tangent method is based on large size estimates, we need to obtain positive
expressions for the various counting functions we wish to estimate, however LU decomposition naturally yields alternating sum expressions 
for these. We therefore have to reformulate these alternating sums as positive sums, a rather involved procedure. 
\item Each setting for applying the Tangent Method only yields a portion of the arctic curve, hence we must change the setting (and possibly
the NILP formulations) to obtain other portions.
\end{itemize}

The paper is organized as follows. 

In the preliminary Section~\ref{prelimsec}, we expose the general Tangent Method, and how to apply it
in the context of NILP models. In particular we describe a general approach to the computation of partition functions of NILP
and their version with one escaping path, based on the LU decomposition of the corresponding Gessel-Viennot matrix. 
We conclude this section with a review of how to apply generating function methods to this procedure.

Section~\ref{aztecsec} is devoted to the re-derivation of the arctic circle for domino tilings of the Aztec diamond by the tangent method
applied to their reformulation in terms of (large Schr\"oder) NILP.

Sections~\ref{dycksecone} and \ref{dycksectwo} explore another NILP problem involving non-intersecting 
Dyck paths i.e. directed paths on the square lattice
that remain in a half-plane. In Section~\ref{dycksecone}, we explore and solve the path problem and derive the naturally associated portion 
of arctic curve. To get the entire curve, we first reformulate the path problem as a rhombus tiling problem of a cut hexagon in 
Section~\ref{dycksectwo},
where we derive the rest of the arctic curve by considering an alternative NILP description of the same tiling configurations.
The complete result for the arctic curve is a half-ellipse inscribed in the cut hexagon.

In Section \ref{everyothersec}, we address yet another NILP problem involving ordinary directed paths on the square lattice, but with specific 
constraints on their starting and ending points. For this case we find that the arctic curve is a portion of parabola.

Section~\ref{sec:VSASM} is devoted to the derivation of the arctic curve for VSASMs. We find essentially the same result as for 
ordinary ASMs (which do not have the reflection symmetry constraint).

We gather a few concluding remarks in Section~\ref{conclusec}.

\medskip

\noindent{\bf Acknowledgments.} We are thankful to F. Colomo and A. Sportiello for extensive discussions on the tangent method. 
PDF thanks the organizers of the program ``Combi17: Combinatorics and interactions" held at the Institut Henri Poincar\'e, Paris
where the present work originated. MFL would like to acknowledge the support of his advisor Taylor Hughes, and support from
NSF grant DMR 1351895-CAR.
MFL thanks the Galileo Galilei Institute in Florence for hospitality during the 2017 ``Lectures on Statistical Field Theories"
winter school program where the first stages of the present project were carried out, and also the organizers of the 2017
``Exact methods in low-dimensional physics" summer school at Institut d'Etudes Scientifiques in Carg\`ese.  PDF is partially 
supported by the Morris and Gertrude Fine endowment.

\section{Preliminaries}\label{prelimsec}

\subsection{The tangent method}\label{tansec}

The method we are going to describe was devised by Colomo and Sportiello \cite{COSPO}
and applies to a number of problems. First of all tiling problems of plane domains by means of tiles
with a few specific shapes and sizes. In many cases, such problems may be reformulated in terms of Non-Intersecting Lattice Paths (NILP),
for which the tangent method is well-defined. The latter are of course combinatorial problems in their own right, a few of which we will 
consider in this paper. A more subtle application of the method seems to indicate that it also applies to {\it interacting} lattice paths,
typically allowed to ``kiss" at a vertex, with a particular interaction weight. This is the case for the ``osculating path" formulation 
of the configurations of the six vertex model \cite{COSPO}.

Let us consider configurations of families of non-intersecting (directed) lattice paths, say from a set of initial points
$v_0,v_1,...,v_n$ to a set of endpoints ${\tilde v}_0,{\tilde v}_1,...,{\tilde v}_n$. By lattice paths, we mean paths drawn on a directed graph, 
whose vertices are the lattice points and whose elementary oriented steps are taken along  oriented edges of the graph.

The tangent method allows to study the large size asymptotics of those configurations  (in a sense defined below). 
Note that the set of all possible paths from any of the $v$'s to any of the $\tilde v$'s defines a maximal domain $D$
with a shape depending on the lattice and on the positions of the $v,\tilde v$'s. By large size asymptotics, we mean the 
limit of a large (scaled) domain $D$. 

For large size, we expect the configurations to display the following pattern of behavior. As long as the paths remain too close to each other, 
they will form very regular ``crystalline" patterns, whereas far enough from the edges of $D$ we expect more disorder, i.e. a ``liquid" phase.
It turns out that in many NILP problems the separation between these phases is very sharp, giving rise asymptotically to separating curves.
This is the ``Arctic Curve Phenomenon". This can be translated back in terms of tilings by saying that tilings of a very large planar domain $D$
tend to develop crystalline phases (in particular induced by corners) along the boundaries of $D$, while a liquid phase develops away from 
the boundaries. These are separated by Arctic Curves as well.  There is a large literature on this subject originating in the determination of
the ``Arctic Circle" for the domino tiling of the so-called Aztec Diamond \cite{CEP,JPS}, and evolving to general descriptions such as that of \cite{KO1,KOS},
and further developments studying the fluctuations around the arctic curve as well.

\begin{figure}
\centering
\includegraphics[width=10.cm]{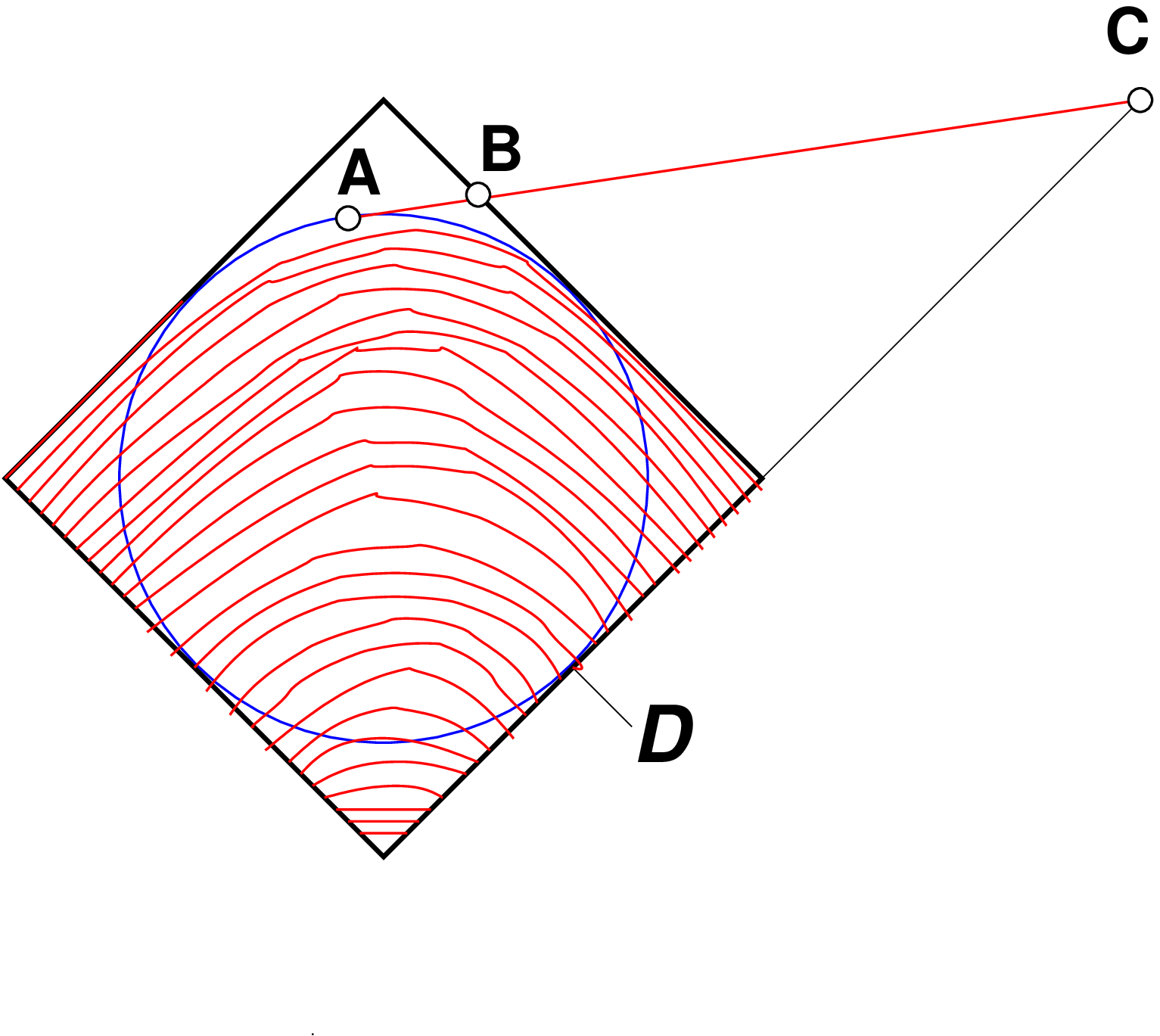}
\caption{\small The tangent method. We have represented a typical NILP problem, with starting points at the left and ending points at the right lower boundaries of a square-shaped domain $D$. The topmost path is constrained to end at the point $C$. The corresponding tangent to the arctic curve is the line $BC$, where $B$ is the most likely exit position from the domain $D$. We have also represented the tangency point $A$ which lies on the arctic curve.}
\label{fig:tangent}
\end{figure}

The tangent method allows to predict the precise shape of the arctic curve and is sketched on a particular example in 
Fig.~\ref{fig:tangent}. 
It goes as follows. The arctic curve can be explored by considering the 
most likely outermost paths of the configuration (say paths from $v_n$ to ${\tilde v}_n$), which in the large size limit will tend to portions of the arctic curve. 
To capture the exact position of the arctic curve, consider another related NILP problem (see Fig.~\ref{fig:tangent}): move the endpoint 
${\tilde v}_n$ to a reachable point ${\tilde v}$ (e.g. the point $C$ in Fig.~\ref{fig:tangent})
far away from the other endpoints, in particular outside of $D$. The corresponding ``outer path" $v_n\to {\tilde v}$ will therefore veer away from 
the other paths, and its most likely
trajectory will be a straight line, when away from the influence of the other paths. We note that this line must be asymptotically  tangent to the arctic curve.
By moving around the position of the new endpoint ${\tilde v}$, we may thus generate some parametric family of tangents to the arctic curve, which can be recovered 
as their envelope.

How to determine the family of tangents? We may record the point ${\tilde w}$ (e.g. the point $B$ on Fig.~\ref{fig:tangent}) from which the escaping path 
$v_n\to {\tilde v}$ leaves the domain $D$
(i.e. the last visit to a boundary point of $D$, and unique vertex from which a step outside of $D$ is taken). 
This point is supposed to be already away from the influence of the other paths. So from this point on,
the escaping paths will follow the tangent asymptotically, defined as the line through ${\tilde w}$ and ${\tilde v}$ (e.g. the line $BC$ in Fig.~\ref{fig:tangent}). 
Assume we can enumerate the NILP configurations
from $v_0,v_1,...,v_n$ to ${\tilde v}_0,{\tilde v}_1,...,{\tilde v}_{n-1},{\tilde v}$, 
say with partition or counting function $N_{\tilde v}$, as well as those 
from $v_0,v_1,...,v_n$ to ${\tilde v}_0,{\tilde v}_1,...,{\tilde v}_{n-1},{\tilde w}$, with partition function $Z_{\tilde w}$.
Then we can perform an asymptotic large size analysis of $N_{\tilde v}=\sum_{{\tilde w}\in \partial D} Z_{\tilde w} \, Y_{{\tilde w},{\tilde v}}$,
where $Y_{{\tilde w},{\tilde v}}$ is the partition function for paths from $\tilde w$ to $\tilde v$ exiting from $D$ at ${\tilde w}$.
Assuming scaling behaviors ${\tilde v}=n a$, ${\tilde w}=n x$, $n\to \infty$, $a,x$ two points independent of $n$,
we get the following large $n$ leading behavior:
$N_{n a}\sim \int d^2x\ e^{n S(x)}$, for some action functional $S$. By a saddle-point analysis, we find the most likely 
scaled exit point $x^*$, as a function of $a$. This gives a parametric family of tangents, i.e. the lines through
the scaled points $x^*(a)$ and $a$, where $a$ is chosen arbitrarily (so that $n a$ is outside of the scaled domain $D$).

In order to justify the tangent method, we need the following.

\begin{thm}
Let us consider the set $P_{a,b}$ of weighted paths from the origin $(0,0)$ to a point $(a,b)$ on the $\Z^2$ lattice, that use a finite family of steps 
$\{(s_i,t_i),\ i=1,2,...,k\}$ 
with $t_i>0$ for all $i$ (directed paths moving to the right), with respective weights $w_i$, $i=1,2,...,k$. 
In this model each path $p\in P_{a,b}$ is weighted by $w(p)=\prod_{{\rm steps}\  s\ {\rm of}\ p} w(s)$, where $w(s)=w_i$ if $s=(s_i,t_i)$.
In other words, defining the partition function of the model as
$$Z_{(0,0)\to (a,b)}=\sum_{p\in P_{a,b}} w(p)$$
we consider the paths of $P_{a,b}$ as a statistical ensemble with probability weights $P(p)=w(p)/Z_{(0,0)\to (a,b)}$.
Then the most likely configuration of a weighted path in 
$P_{n\al,n\beta}$ for large $n$ is the straight line from the origin to $(a,b)=n(\al,\beta)$.
\end{thm}
\begin{proof}
Let $P$ denote the Newton polynomial of the collection of weighted steps:
$$P(x,y):=\sum_{i=1}^k w_i\, x^{s_i}\, y^{t_i}$$
Then the generating partition function for weighted paths from $(0,0)$ to $(a,b)$, with an additional weight $t$ per step reads:
$$f_{a,b}(t)=\sum_{{\rm paths}\ p\in P_{a,b}} w(p) t^{\ell(p)}=\frac{1}{1-t P(x,y)}\Bigg\vert_{x^ay^b}=\oint \frac{dx dy }{(2i\pi)^2 x^{a+1} y^{b+1}}
\frac{1}{1-t P(x,y)}
$$
where $\ell(p)$ is the length of $p$, i.e. its total number of steps, and 
where we use the notation $f(x,y)\vert_{x^ay^b}$ for the coefficient of $x^ay^b$ in the series $f(x,y)$ (here we expand the fraction as
a series of $t$).  
We may split the corresponding partition function $Z_{(0,0)\to (a,b)}$ into two pieces by recording the position of 
some intermediate point $(x,y)$ in a domain $L$, resulting in the decomposition:
$$Z_{(0,0)\to (a,b)}=\sum_{(x,y)\in L}Z_{(0,0)\to (x,y)}\, Z_{(x,y)\to (a,b)}$$
Let us now consider the scaling limit where $(a,b)=n(\al,\beta)$, and $(x,y)=n(\xi, \eta)$, $L=n{\mathcal L}$, with very large $n$.
We have asymptotically
$$Z_{(0,0)\to n(\al,\beta)}\sim \int_{\mathcal L} d\xi d\eta\   Z_{(0,0)\to n(\xi, \eta)} Z_{n(\xi, \eta)\to n(\al,\beta)} $$
where 
$$Z_{(0,0)\to n(\al,\beta)}=\oint \frac{dx dy }{(2i\pi)^2 x y} e^{n S_{\al,\beta}(x,y)},\quad 
S_{\al,\beta}(x,y)=-\al {\rm Log}(x)-\beta{\rm Log}(y)-\frac{1}{n}{\rm Log}(1-t P(x,y))$$
and 
$$
Z_{n(\xi, \eta)\to n(\al,\beta)}=Z_{(0,0)\to n(\al-\xi,\beta-\eta)}
$$
by translational invariance of the problem.
We may rewrite
$$
Z_{(0,0)\to n(\al,\beta)}\sim \int d\xi d\eta \oint \frac{dxdydudv}{(2i\pi)^4 x y u v}
e^{n S}, \quad S=S_{\xi,\eta}(x,y)+S_{\al-\xi,\beta-\eta}(u,v)
$$
The integral is dominated by the extrema of the action. Writing $\partial_\xi S=\partial_\eta S=0$ implies $x=u$ and $y=v$.
Moreover we have:
\begin{eqnarray*}
\partial_x S&=& -\frac{\xi}{x}+\frac{1}{n} \frac{t\partial_x P(x,y)}{1-t P(x,y)} =0 \\
\partial_y S&=& -\frac{\eta}{y}+\frac{1}{n} \frac{t\partial_y P(x,y)}{1-t P(x,y)} =0 \\
\partial_u S&=& -\frac{\al-\xi}{u}+\frac{1}{n} \frac{t\partial_u P(u,v)}{1-t P(u,v)} =0 \\
\partial_v S&=& -\frac{\beta-\eta}{v}+\frac{1}{n} \frac{t\partial_v P(u,v)}{1-t P(u,v)} =0
\end{eqnarray*}
At the saddle-point, as $u=x$ and $v=y$, we deduce that:
$$ 
\frac{\xi}{\eta}=\frac{\al-\xi}{\beta-\eta}\ \Rightarrow \ \beta\xi-\al\eta=0 
$$
We conclude that the most likely intermediate point lies on the line through $(0,0)$ and $(\al,\beta)$, and the theorem follows.
\end{proof}

To make the tangent method completely rigorous, we should argue that the path is most likely to follow a straight line {\it until} 
the contact with the arctic curve, when the influence of the other paths starts to play a role. One also should argue that the reasoning 
holds irrespectively of possibly stronger interactions between paths (like in the six vertex case). 

\subsection{Non-intersecting paths and $LU$ method}

\label{sec:paths-and-LU}

Most of the models that we consider in this paper, with the exception of the Vertically Symmetric Alternating Sign Matrices
in Section~\ref{sec:VSASM}, have a formulation in terms of 
NILPs on a directed graph $\mathcal{G}$. The general setup is as follows. To start,
we may allow the paths to intersect, and we only impose the non-intersecting condition after setting down the
general definitions. We consider $n+1$ lattice paths on $\mathcal{G}$, with oriented steps along the oriented edges of ${\mathcal G}$,
which begin at the vertices $v_i$ and end at vertices
$\tilde{v}_j$, with $i,j$ integers in the range $0,\dots,n$. For each oriented edge $e$ of the graph we assign a weight 
$w(e)$. To each individual path $p$ we assign the weight $w(p)= \prod_{e\in p} w(e)$. 
For a collection of $n+1$ lattice paths $\mathcal{P}=\{p_0,p_1,\dots,p_n\}$ we assign the weight
$w(\mathcal{P})= \prod_{j=0}^n w(p_j)$. Finally, we define the partition function for lattice paths on the
graph $\mathcal{G}$ which start at the vertices $\{v_i\}_{i=0}^n$ and end at $\{\tilde{v}_j\}_{j=0}^n$ to be
\beq
	Z(\{v_i\} \to \{\tilde{v}_j\})= \sum_{{\mathcal{P}}:\ \{v_i\} \to \{\tilde{v}_j\}} w({\mathcal{P}})\ .
\eeq
The sum is taken over all sets ${\mathcal{P}}$ of $n+1$ lattice paths on $\mathcal{G}$ which connect the starting vertices 
$v_i$ to the ending vertices $\tilde{v}_j$.

To address the case of non-intersecting lattice paths, we shall use the celebrated Gessel-Viennot formula \cite{GV}. 
However it is only applicable if our path setting satisfies the following crossing property: 

\noindent{}(CP)\ {\it For any $0\leq i<j\leq n$ and $0\leq k<\ell\leq n$, 
any path from $v_i$ to ${\tilde v}_\ell$ and any path from $v_j$ to ${\tilde v}_k$ must intersect at least once}. 

This puts in principle some restriction
on the choice of starting and ending points, depending on the structure of the directed graph ${\mathcal G}$. However for all the applications
in the present paper, the crossing property will always be trivially satisfied.

\begin{lemma}[\cite{GV}]
\label{lemma:LGV}
For non-intersecting lattice paths satisfying the property (CP),
the partition function $Z(\{v_i\} \to \{\tilde{v}_j\})$ is given by the determinant formula
\beq
	Z(\{v_i\} \to \{\tilde{v}_j\})= \det_{i,j\in[0,n]}\left(Z_{i,j}\right)\ ,
\eeq
where
\beq
	Z_{i,j}= \sum_{p:\ v_i\to \tilde{v}_j} w(p)\ ,
\eeq
is the weight for all paths from $v_i$ to $\tilde{v}_j$.
\end{lemma}


In applying the tangent method to non-intersecting path models we encounter the following situation. 
We start with the original model of non-intersecting paths on a certain domain $D$ of the directed lattice graph ${\mathcal G}$.
Its  partition function $Z$ is given by Lemma~\ref{lemma:LGV} as the determinant of a certain matrix $A$, $Z=\det(A)$. 

Next, we must consider the partition function $N_{\tilde v}$ of the same model, but with the endpoint ${\tilde v}_n$ 
moved to a different location ${\tilde v}$ outside of $D$, such that the Gessel-Viennot formula of Lemma \ref{lemma:LGV} still applies.
This partition function may be decomposed according to the position ${\tilde w}$ (on the boundary $\partial D$ of $D$)
of the exit point from the domain $D$ of the path from $v_n\to {\tilde v}$ (we choose the point $\tilde v$ so that the path only 
exits once, and can never visit $D$ again). Let $Z_{\tilde w}$ denote the partition function of paths on $D$ from $\{v_0,v_1,...,v_{n}\}$
to  $\{{\tilde v}_0,{\tilde v}_1,...,{\tilde v}_{n-1},{\tilde w}\}$, and $Y_{{\tilde w},{\tilde v}}$ that for paths from ${\tilde w}$ to ${\tilde v}$
that exit $D$. Then we have
\beq
N_{\tilde v}=\sum_{{\tilde w} \in \partial D} Z_{\tilde w}\, Y_{{\tilde w},{\tilde v}} 
\eeq
Applying the Gessel-Viennot formula of Lemma~\ref{lemma:LGV}, we find that $Z_{\tilde w}$ is the
determinant of a matrix $A^{(\tilde w)}$ which differs from $A$ only in its last column, $Z_{\tilde w}=\det(A^{(\tilde w)})$. 
For each ${\tilde w}$ we have
\beq
	A^{({\tilde w})}_{i,j} =\begin{cases}
	A_{i,j} &,\ j\in[0,n-1] \\
	b^{({\tilde w})}_i &,\ j=n
\end{cases}\ ,
\eeq
where 
$A_{i,j}=Z_{v_i\to {\tilde v}_j}$ and $b^{({\tilde w})}_i=Z_{v_i\to {\tilde w}}$ is the partition function for the paths from 
$v_i$ to ${\tilde w}$ on $D$.

To compute the determinant of $A^{({\tilde w})}$ we use the $LU$ decomposition. First, let us suppose that the original
matrix $A$ has an $LU$ decomposition as
\beq
	A= LU\ ,
\eeq
where $L$ is a lower triangular matrix with $1$'s on the diagonal and $U$ is an upper triangular matrix. Then we have
\beq
	\det_{i,j\in[0,n]}\left(A_{i,j}\right)= \prod_{i=0}^n U_{i,i}\ .
\eeq
In the cases we consider in this paper the $LU$ decomposition of the matrix $A$ is either known and can be found in the
literature, or we compute the matrices $L$ and $U$ explicitly using various methods.
This gives immediately an $LU$ decomposition for the matrices $A^{({\tilde w})}$ as well:

\beq
	A^{({\tilde w})}= LU^{({\tilde w})}\ ,
\eeq
where
\beq
	U^{({\tilde w})}:= L^{-1}A^{({\tilde w})}\ .
\eeq
It then follows that
\beq
	\det_{i,j\in[0,n]}\left(A^{({\tilde w})}_{i,j}\right)= \prod_{i=0}^n U^{({\tilde w})}_{i,i}\ .
\eeq

The advantage of this method is that a major simplification occurs due to the fact that $A^{({\tilde w})}$ differs from $A$ only
in the last column. Indeed, we have $U^{({\tilde w})}_{i,j}= \sum_{k=0}^n (L^{-1})_{i,k}A^{({\tilde w})}_{k,j}$, but 
$A^{({\tilde w})}_{k,j}= A_{k,j}$ for $j<n$. This means that for $j<n$, we have $U^{({\tilde w})}_{i,j}= U_{i,j}$. Thus, we find that the 
determinant of $A^{({\tilde w})}$ is given by
\beq
\det_{i,j\in[0,n]}\left(A^{({\tilde w})}_{i,j}\right) = \left(\prod_{i=0}^{n-1} U_{i,i}\right) U^{({\tilde w})}_{n,n}\ .
\eeq
In other words, the computation of the determinant of $A^{({\tilde w})}$ reduces to the computation of the single matrix element
\beq
U^{({\tilde w})}_{n,n} = \sum_{k=0}^n(L^{-1})_{n,k}A^{({\tilde w})}_{k,n}
= \sum_{k=0}^n(L^{-1})_{n,k}b^{({\tilde w})}_k\ .
\eeq
In the context of the tangent method it is more useful to consider a ``one-point function" which is defined as the following
ratio $H_{\tilde w}$ of partition functions:
\beq
H_{\tilde w}:= \frac{Z_{\tilde w}}{Z}= \frac{\det(A^{({\tilde w})})}{\det(A)}
\eeq
Using the $LU$ decomposition method we find that this
function is given by
\beq
H_{\tilde w}=\frac{U^{({\tilde w})}_{n,n}}{U_{n,n}}\ .
\eeq

This gives an efficient way for calculating $H_{\tilde w}$. The tangent method can then be applied to the decomposition of the
normalized partition function, in the limit of a large scaled domain $D$:
\beq
\frac{N_{\tilde v}}{Z}=\sum_{\tilde w} H_{\tilde w} \, Y_{{\tilde w},{\tilde v}} \ .
\eeq
However, to perform an asymptotic analysis of this sum, we need to be able to estimate $H_{\tilde w}$ in large size.
Note that the expression for $U^{({\tilde w})}_{n,n}$ obtained from the computation of $L^{-1}A^{({\tilde w})}$ is naturally
an {\it alternating} sum due to the cofactor expansion of $L^{-1}$. Such sums are not suitable for asymptotic analysis,
as large terms may cancel out. We will therefore need to reexpress $H_{\tilde w}$ as positive sums, a technically
demanding process.

\subsection{Working with generating functions}

In the core of the paper, we will use a number of manipulations involving generating functions 
which we gather here as a preliminary.

\subsubsection{Generating functions and binomial identities}

For a series or polynomial $f(x)=\sum_{i\geq 0} f_i x^i$ we write $f_n=f(x)\vert_{x^n}$ for the coefficient of $x^n$ in $f(x)$ 
(in particular when $n=0$ this denotes the
{\it constant term} $f_0$ of $f$). We also have by the Cauchy theorem:
\beq f(x)\vert_{x^n}=\oint \frac{dx}{2i\pi x^{n+1}} f(x)
\eeq
where the contour of integration is around $0$, and the integral picks up the residue at $0$.

The following lemma will be used throughout the paper.

\begin{lemma}\label{prodlem}
For any two generating functions $f(x),g(x)$, we have the following identity:
\beq 
f(x)g(x)\vert_{x^k}=\sum_{\ell=0}^k f(x)\vert_{x^\ell}\, g(x)\vert_{x^{k-\ell}}
\eeq
\end{lemma}
\begin{proof}
This expresses simply $(fg)_n$, the coefficient of $x^n$ in $fg$, as $(fg)_n=\sum_{\ell=0}^k f_\ell\, g_{k-\ell}$.
\end{proof}

In this paper, we will deal with expressions involving typically sums of products of binomial coefficients. 
The following lemma will be used repeatedly.

\begin{lemma}\label{fourwaylem}
We have four simple
ways of expressing the binomial coefficient ${n\choose k}$ as the coefficient in a generating series or polynomial:
\beq{n\choose k}=(1+x)^n\vert_{x^k}=(1+x)^n\vert_{x^{n-k}}=\frac{1}{(1-x)^{k+1}}\Bigg\vert_{x^{n-k}}=\frac{1}{(1-x)^{n-k+1}}\Bigg\vert_{x^{k}}
\eeq
\end{lemma}
\begin{proof} The two first expressions are obvious. To get the third and fourth, simply notice that
\beq\frac{1}{(1-x)^{m+1}}=\sum_{j\geq 0} {m+j\choose j} x^j \eeq
\end{proof}

We shall also use the description below, using iterated derivatives.
\begin{lemma}\label{iterdetlem}
We have the following representations of the binomial coefficient ${n\choose k}$:
\beq{n\choose k}=\frac{1}{k!}\left(\frac{d^k}{dt^k} t^n \right)\Bigg\vert_{t=1}=
\frac{1}{(n-k)!}\left(\frac{d^{n-k}}{dt^{n-k}} t^n \right)\Bigg\vert_{t=1}
\eeq
\end{lemma}

The collection of lemmas above will be used extensively throughout this paper
to evaluate summations over products of binomial coefficients.

A last formula concerns the description of the {\it inverse} of a binomial coefficient, by means of a hypergeometric series.
By these we mean precisely the series
$${}_{2}F_{1}(a,b,c;x):=\sum_{n=0}^\infty \frac{(a)_n\, (b)_n}{(c)_n \, n!} \, x^n=\int_{0}^1 dt \, t^{b-1} (1-t)^{c-b-1} (1-t x)^{-a}$$
We have the following direct application:
\begin{lemma}\label{hypergeom}
The inverse of the binomial coefficient ${n+a\choose a}$ for $n,a\geq 0$  is given by:
$${n+a\choose a}^{-1}=a\times {}_{2}F_{1}(1,1,a+1;x)\big\vert_{x^n}$$
\end{lemma}

Using the explicit integral formulation above, we may also write:
\begin{eqnarray*}
\sum_{n=0}^\infty {n+a\choose a}^{-1}x^n&=&a\int_{0}^1 dt \, \frac{(1-t)^{a-1}}{ 1-t x}=a\int_0^1 dt\, \frac{t^{a-1}}{1-x+t x}\\
&=&a\sum_{m=0}^\infty (-1)^m \frac{x^m}{(1-x)^{m+1}}\int_0^1dt\, t^{m+a-1}=\sum_{m=0}^\infty \frac{(-1)^m\, a}{m+a} \frac{x^m}{(1-x)^{m+1}}
\end{eqnarray*}
Extracting the coefficient of $x^n$ we get the following formula:
\begin{lemma}
The inverse of the binomial coefficient ${n+a\choose a}$ for $n,a\geq 0$  is given by:
$$ {n+a\choose a}^{-1}=\sum_{m=0}^n \frac{(-1)^m\, a}{m+a} \, {n\choose m} $$
\end{lemma}

\subsubsection{Infinite matrices and their truncations}\label{truncsec}

In this paper, we will also deal with some finite truncations of infinite matrices, whose manipulation is greatly simplified
by use of generating functions. For an infinite matrix $A$ with entries $A_{i,j}$, $i,j\in \Z_{\geq 0}$, we introduce the
generating function
$$f_A(z,w):=\sum_{i,j\geq 0} z^i w^j\, A_{i,j} $$
Then we have the following result.
\begin{lemma}\label{convolem}
For any two infinite matrices $A,B$ with generating functions $f_A,f_B$, assuming the product $AB$ makes sense,
we have the following product formula:
$$ f_{AB}(z,w)=f_A\star f_B(z,w)=f_A(z,t^{-1})f_B(t,w)\vert_{t^0}=\oint \frac{dt}{2i\pi t} f_A(z,t^{-1})f_B(t,w)$$
where $\star$ stands for the convolution product of two-variable generating functions, and the contour integral is 
for instance over the unit circle.
\end{lemma}
\begin{proof}
We write:
$$ f_{AB}(z,w)=\sum_{i,j\geq 0} z^i w^j \sum_{r\geq 0} A_{i,r}B_{r,j}=f_A(z,t^{-1})f_B(t,w)\vert_{t^0}$$
\end{proof}

Another important property of infinite matrices is that the $LU$ factorization of an infinite matrix $A$ descends to its
finite truncations $A^{(n)}=(A_{i,j})_{i,j\in [0,n]}$. More precisely, we have the following:
\begin{lemma}\label{truncLUlem}
Let $A$ be an infinite matrix. Assume there exist infinite matrices $L,U$ respectively uni-lower- and upper-triangular
such that $A=LU$. Then for all $n\geq 0$ we have the following $LU$ factorization:
\beq\label{truncA} A^{(n)}=L^{(n)}U^{(n)} \eeq
and moreover if $L$ is invertible, its inverse also truncates:
\beq\label{truncLU} (L^{(n)})^{-1}=(L^{-1})^{(n)} \eeq
\end{lemma}
\begin{proof}
Assuming $A=LU$, we compute:
$$A^{(n)}=(\sum_{r\geq 0} L_{i,r}U_{r,j})_{i,j\in [0,n]}=(\sum_{r=j}^i L_{i,r}U_{r,j})_{i,j\in [0,n]}
=(\sum_{r=j}^i L_{i,r}^{(n)}U_{r,j}^{(n)})_{i,j\in [0,n]}$$
where the truncation appears {\it de facto} from triangularity and \eqref{truncA} follows.
Assuming $L^{-1}$ exists, it is also uni-triangular, and we write similarly:
\begin{eqnarray*}
I^{(n)}&=&(L^{-1}L)^{(n)}=(\sum_{r\geq 0} (L^{-1})_{i,r}L_{r,j})_{i,j\in [0,n]}=(\sum_{r\geq 0}^i (L^{-1})_{i,r}L_{r,j})_{i,j\in [0,n]}\\
&=&
(\sum_{r\geq 0}^i ((L^{-1})^{(n)})_{i,r}(L^{(n)})_{r,j})_{i,j\in [0,n]}
\end{eqnarray*}
which implies \eqref{truncLU}
\end{proof}

\section{Aztec Diamond domino tilings}\label{aztecsec}

\begin{figure}
\centering
\includegraphics[width=15.cm]{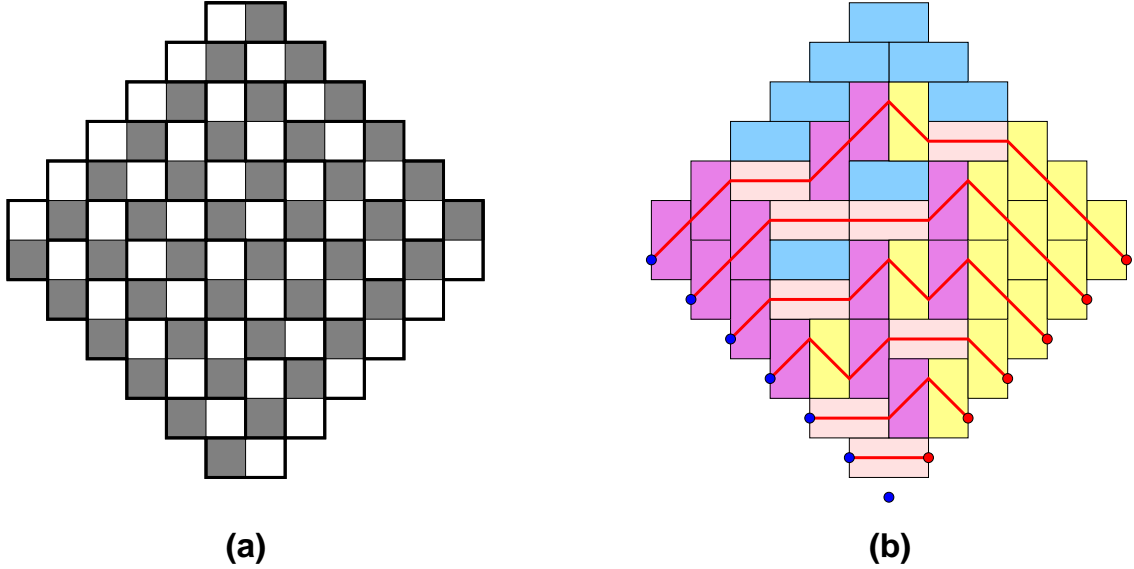}
\caption{\small The Aztec diamond domino tiling problem and its path formulation. 
(a) The face bi-colored Aztec diamond is represented in thick black line, 
together with a sample domino tiling. The size $n$ is the number
of black squares on it SW boundary ($n=6$ here). (b) The equivalent non-intersecting large Schroeder path configuration is
represented in red. We have singled out the four types of dominos with four different colors.}
\label{fig:aztec}
\end{figure}

\subsection{Path formulation}

The Aztec diamond domino tiling problem can be stated as follows: tile the size $n$ domain $D$ of the square lattice $\Z^2$
depicted in Fig.~\ref{fig:aztec} by means of dominos made of two elementary squares glued along a common edge.
After bi-coloring the square faces of $D$ (in black and white), so that South-West (SW) boundary squares are colored
in black, we see that there are $4$ distinct bi-colored tiles compatible with the bi-coloration of the faces of $D$. 

Domino tiling configurations of $D$ are in bijection with non-intersecting paths. Simply associate a portion of path drawn on the 
four possible domino tiles:
\begin{equation}\label{dominos} \raisebox{-1.cm}{\hbox{\epsfxsize=14.cm \epsfbox{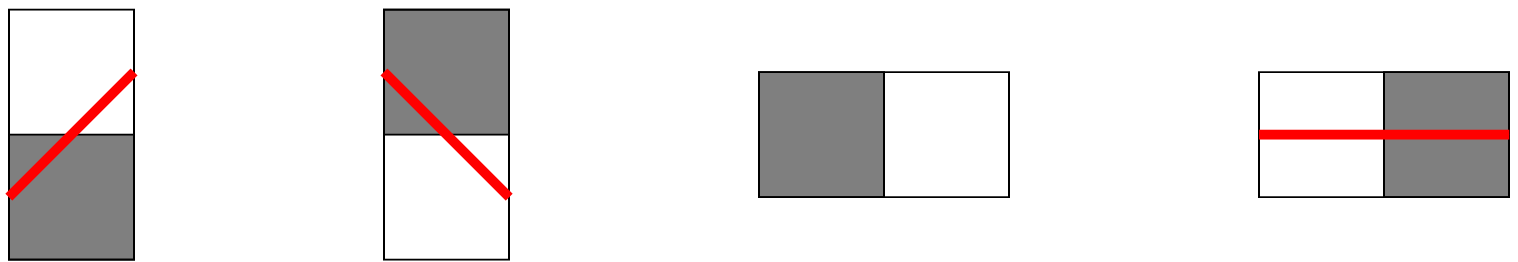}}} 
\end{equation}
These paths join vertices at the middle of edges of the original square lattice, and are non-intersecting by definition. Note that 
some regions may not be visited by paths, due to the fact that the third domino above is not crossed by any path.
For the Aztec diamond of size $n$, the domino tiling configurations are equivalent to those of $n+1$ NILPs with steps $(1,1)$, $(1,-1)$ and $(2,0)$ which we call respectively up, down and horizontal,
corresponding to the 1st, 2nd and 4th dominos above.
Moreover $n$ of the paths start at the middle of the West edge of the SW boundary black square faces introduced above.
We may choose coordinates in which these are $(-i,i)$, $i=1,2,...,n$. These paths end on the SE boundary white square faces
at points with coordinates $(j,j)$, $j=1,2,...,n$. In addition, we consider the origin to be both the starting and end point of a trivial path
of length $0$. The non-intersection constraint then forces any path from $(-1,1)$ to $(1,1)$ to have either a horizontal step
or a succession up-down, but prevents the unwanted down-up succession. 
Such paths are usually referred to as {\it Large Schroeder Paths}. 
If we assign a weight of $1$ to each edge of the lattice, then the partition function of this
system is equal to the number $2^{n(n+1)/2}$ of domino tilings of the Aztec diamond of size $n$. We now
briefly review this result.

\subsection{Exact enumeration}

We now review the calculation, in the paths formulation, of the number of domino tilings of the Aztec diamond. First, 
if we assign a weight of $1$ to each edge of the lattice, then the partition function for all paths 
from $(-i,i)$ to $(j,j)$ is readily calculated to be 
\beq
	Z_{(-i,i)\to(j,j)}=A_{i,j}:= \sum_{p=0}^{\text{min}(i,j)} \binom{i+j-p}{p,i-p,j-p}\ . \label{eq:Aztec-weights}
\eeq
where in the summation $p$ denotes the number of horizontal steps, and $i-p,j-p$ the numbers of up and down steps, respectively.
Then by Lemma~\ref{lemma:LGV}, the partition function for this system is $Z=\text{det}(A)$. This determinant can 
be computed with the help of the following Lemma. 

\begin{lemma}
\label{lemma:aztec-LU}
The matrix $A$ admits an LU decomposition $A= LU$ with
\beq
	L_{i,j}= \binom{i}{j}
\eeq
and
\beq
	U_{i,j}= 2^i \binom{j}{i}\ .
\eeq
\end{lemma}
Since $L_{i,i}=1$ and $U_{i,i}= 2^i$, an immediate consequence of the Lemma is that
\beq\label{aztecpart}
	\det_{i,j\in[0,n]}(A_{i,j})= \prod_{i=0}^n U_{i,i}= 2^{\frac{n(n+1)}{2}}\ ,
\eeq
which is exactly equal to the number of domino tilings of the Aztec diamond of size $n$.

\begin{proof}
To prove Lemma~\ref{lemma:aztec-LU} we use the infinite matrix generating function methods 
of Section \ref{truncsec}. 
More precisely, we find a simple $LU$ decomposition of the infinite matrix $A=(A_{i,j})_{i,j\in \Z_{\geq 0}}$,
with generating function:
\beq\label{Agen}
	f_A(z,w):= \sum_{i,j=0}^{\infty}A_{i,j}z^i w^j= \frac{1}{1-z-w-zw}\ .
\eeq
This is proved by working backwards from the generating function. We have
\beqa
	 \frac{1}{1-z-w-zw} &=& \sum_{m=0}^{\infty}(z+w+wz)^m \nnb \\
	&=& \sum_{m=0}^{\infty}\underset{ p_1+p_2+p_3=m}{\sum_{p_1,p_2,p_3}} \binom{m}{p_1,\ p_2,\ p_3} z^{p_1}w^{p_2}(zw)^{p_3} \nnb \\
	&=& \sum_{p_1,p_2,p_3=0}^{\infty}\binom{p_1+p_2+p_3}{p_1,\ p_2,\ p_3}z^{p_1+p_3}w^{p_2+p_3} \ ,\nnb
\eeqa
by the standard trinomial identity. To extract
the coefficient of $z^i w^j$ in this we set $p_1= i-p_3$, $p_2= j-p_3$ to find:
\beq
	 \frac{1}{1-z-w-zw} \Big|_{z^i w^j}= \sum_{p_3=0}^{\text{min}(i,j)}\binom{i+j-p_3}{i-p_3,\ j-p_3,\ p_3}\ ,
\eeq
which coincides with our original expression for the weights $A_{i,j}$.

A particular $LU$ decomposition of $A$ is obtained by checking that:
\beq
	f_A(z,w)= (f_L\star f_U)(z,w)
\eeq
where
\beq
	f_L(z,w) = \frac{1}{1-z-zw} \quad {\rm and}\quad 
	f_U(z,w) = \frac{1}{1-w-2zw}\ .
\eeq
Indeed, we have:
$$(f_L\star f_U)(z,w)=\oint \frac{dt}{2i\pi} \frac{1}{t(1-z)-z}\frac{1}{1-w-2 t w}=\frac{1}{1-z}\frac{1}{1-w-2 \frac{z}{1-z} w}
=\frac{1}{1-z-w-zw}$$
where we noted that the contour integral picks up the residue at $t=z/(1-z)$.
Comparing this with \eqref{Agen} implies that the matrix $A$ factorizes as $A=LU$ where the $(i,j)$ matrix elements of $L$ and $U$ are given by
the coefficient of $z^i w^j$ in the expansion of the generating functions $f_L(z,w)$ and $f_U(z,w)$ about $z=w=0$. 
Explicitly, we have
\beq
	f_L(z,w) = \sum_{i=0}^{\infty} (1+w)^i z^i = \sum_{i=0}^{\infty}\sum_{j=0}^i \binom{i}{j}z^i w^j\ ,
\eeq
from which we find that $L$ is lower uni-triangular, with entries $L_{i,j}= \binom{i}{j}$. Similarly, we have: 
\beq
f_U(z,w)=\sum_{j=0}^{\infty} (1+2z)^j w^j=\sum_{j=0}^\infty\sum_{i=0}^j 2^{i}{j\choose i} z^i w^j
\eeq 
from which we get that $U$ is upper triangular with entries $U_{i,j}=2^{i}{j\choose i}$.
As noted in Section \ref{truncsec}, this $LU$ decomposition holds for the finite matrices obtained by 
truncation to indices $i,j\in [0,n]$,
and Lemma~\ref{lemma:aztec-LU} follows (Note that we dropped the superscript $(n)$ in $A,L,U$ for simplicity).
\end{proof}

For later use we note that the infinite matrix $L$ is invertible, and the inverse matrix $L^{-1}$ has matrix elements
\beq
	(L^{-1})_{i,j}=  (-1)^{i+j}\binom{i}{j} \label{eq:L-inv-aztec}
\eeq
and generating function 
\beq
	f_{L^{-1}}(z,w)= \frac{1}{1+z-zw}\ .
\eeq
Indeed, we easily compute:
$$(f_{L^{-1}}\star f_L)(z,w)=\oint \frac{dt}{2i\pi} \frac{1}{t(1+z)-z}\frac{1}{1-t(1+w)}
=\frac{1}{1+z}\frac{1}{1-\frac{z(1+w)}{1+z}}=\frac{1}{1-z w}=f_{\mathbb I}(z,w)$$
where the contour integral has picked the residue at $t=z/(1+z)$.
By  Lemma \ref{truncLUlem}, eq.\eqref{eq:L-inv-aztec} also holds for the finite truncation of $L$ to indices
$i,j\in [0,n]$.

As explained in Section \ref{tansec}, to set up the tangent method for the Aztec diamond we extend the original 
diamond shaped domain by appending a
rectangular region with corners at $(n,n)$, $(0,2n)$, $(k-n,k+n)$ and $(k,k)$, for some integer $k>n$. We again consider
$n+1$ paths on this domain with the same starting points $(-i,i)$, $i\in[0,n]$ as before, but now we take the end points to 
be $(j,j)$, $j\in[0,n-1]$, and $(k,k)$. Thus, in the non-intersecting configuration the top-most path will start at $(-n,n)$ and end
at $(k,k)$ at the far upper right corner of the extended domain. 

Let us denote the partition function for these $n+1$ paths on the extended domain by $N(n,k)$. To apply the tangent
method we expand $N(n,k)$ in terms of position $(\ell,2n-\ell)$, $\ell\in[0,n]$, of the point where the top-most path exits 
from the original domain $D$ into the extended domain. We have
\beq\label{Nfunction}
	N(n,k)= \sum_{\ell=0}^{n}Z(n,\ell)Y(\ell,k)\ ,
\eeq
where $Z(n,\ell)$ is the partition function for $n+1$ paths on the original domain with starting points at $(-i,i)$, $i\in[0,n]$
and ending points at $(-j,j)$, $j\in[0,n-1]$ and $(\ell,2n-\ell)$, and $Y(\ell,k)$ is the partition function for the single path
exiting $D$ at $(\ell,2n-\ell)$, and ending at $(k,k)$. In computing $Y(\ell,k)$ we should only include contributions from paths 
which start with a 
step of the form $(1,1)$, or $(2,0)$ (otherwise the path would not leave $D$). It is therefore the sum of two contributions, one
for paths $(\ell+1,2n-\ell+1)\to (k,k)$ and one for paths $(\ell+2,2n-\ell)\to (k,k)$. By translation invariance the paths
from $(a,b)\to (c,d)$ contribute the same as those from $(0,0)\to (c-a,d-b)$, with partition
function $A_{i,j}$ \eqref{eq:Aztec-weights}, where $i$ and $j$ are such that $i+j=c-a$ and $j-i=d-b$, hence $i=\frac{b+c-a-d}{2}$
and $j=\frac{c+d-a-b}{2}$.
As a consequence,
$Y(\ell,k)$ is given by:
\beq\label{Yfunction}
	Y(\ell,k)= A_{n-\ell,k-n-1}+A_{n-\ell-1,k-n-1}
\eeq

The partition function $Z(n,\ell)$ is given by Lemma~\ref{lemma:LGV} as the determinant of a certain $(n+1)\times (n+1)$
matrix $\tilde A$ which is constructed as follows. First note that the weight for all paths from $(-i,i)$ to $(\ell,2n-\ell)$ is
by translational invariance the same as that from $(0,0)\to (\ell+i,2n-\ell-i)$, hence:
\beq
	b^{(\ell)}_i :=Z_{(-i,i)\to (\ell,2n-\ell)}=A_{i+\ell-n,n}= \sum_{p=0}^{\text{min}(\ell+i-n,n)} \binom{\ell+i-p}{p,\ell+i-n-p,n-p}\ .
\eeq
Then $\tilde A$ is equal to the $(n+1)\times(n+1)$ matrix which is identical to $A$ 
(the finite version of $A$ with $n+1$ rows and columns) except for its last column, which is given 
by the weights $b^{(\ell)}_i$, i.e. we have
\beq
	{\tilde A}_{i,j}= \begin{cases}
	A_{i,j} &, j\in[0,n-1] \\
	b^{(\ell)}_i &,\ j=n
\end{cases}\ .
\eeq
Note also that ${\tilde A}|_{\ell=n}= A$. By Lemma~\ref{lemma:LGV}, the partition function for the Aztec diamond with the
last path exiting at $(\ell,2n-\ell)$ instead of ending at $(n,n)$ is equal to $\det_{i,j\in[0,n]}({\tilde A})$. For the application to the
tangent method it is more useful to consider the one point function
\beq
	H(n,\ell):= \frac{\det_{i,j\in[0,n]}({\tilde A})}{\det_{i,j\in[0,n]}(A)}= \frac{Z(n,\ell)}{Z(n,n)}\ .
\eeq
We have the following result.

\begin{thm}
\label{thm:aztec-1pt}
The one point function $H(n,\ell)$, $0\leq \ell\leq n$, for the Aztec diamond tiling problem is 
\beq\label{Hfunction}
	H(n,\ell)= 
	\frac{1}{2^n} \sum_{p=0}^{\ell}\binom{n}{p}  \ .
\eeq
\end{thm}
\begin{proof}
From Section~\ref{sec:paths-and-LU} we have $H(n,\ell)= \frac{{\tilde U}_{n,n}}{U_{n,n}}$, where
${\tilde U}:= L^{-1}{\tilde A}$ and $L^{-1}$ was given in Eq.~\eqref{eq:L-inv-aztec}. In addition, $U_{n,n}= 2^n$,
so we only need to prove that:
\beq
	{\tilde U}_{n,n}=
	\sum_{p=0}^{\ell}\binom{n}{p}  \ . \label{eq:Unnl-aztec}
\eeq
Using generating functions, we may write:
\begin{eqnarray*}
{\tilde U}_{n,n}&=&\sum_{r=0}^n (L^{-1})_{n,r}{\tilde A}_{r,n}=\sum_{r=0}^n (-1)^{n+r}{n\choose r} A_{r+\ell-n,n}
= \sum_{r=0}^n (-1)^{n-r}{n\choose n-r} \frac{1}{1-z-w-zw}\Bigg\vert_{z^{r+\ell-n}w^n} \\
&=&
\sum_{r=0}^n (1-z)^n\big|_{z^{n-r}}\frac{1}{1-z-w-zw}\Bigg\vert_{z^{\ell-(n-r)}w^n}
=\frac{(1-z)^n}{1-z-w-zw}\Bigg\vert_{z^{\ell}w^n}\\
&=&(1-z)^n\frac{(1+z)^n}{(1-z)^{n+1}}\Bigg\vert_{z^\ell}=\frac{(1+z)^n}{1-z}\Bigg\vert_{z^\ell}
=\sum_{r=0}^{\ell}\binom{n}{r} \ .
\end{eqnarray*}
and the theorem follows.
\end{proof}

\subsection{Tangent method}
Following Section \ref{tansec}, we wish to evaluate the function $N(n,k)$ \eqref{Nfunction} in the large size limit
corresponding to the scaling $k=n z$, $n$ large and $z>1$ independent of $n$. The leading large $n$ asymptotics are governed by:
$$\frac{N(n,nz)}{Z(n,n)}\sim \int_0^1 d\xi\ H(n,n\xi) \, Y(n\xi,n z)\ ,$$
where we have set $\ell=n\xi$ and replaced the summation by an integral.
From the explicit formulas \eqref{Yfunction} and \eqref{Hfunction}, and using the Stirling formula, we get the following asymptotic behaviors:
$$ Y(n\xi,n z)\sim 2A_{n(1-\xi),n(z-1)}\sim \int_0^{{\rm Min}(1-\xi,z-1)} d\theta\ e^{n S_0(\xi,\theta,z)}\ ,
$$
where
$$S_0(\xi,\theta,z)=
(z-\xi-\theta){\rm Log}(z-\xi-\theta)-\theta{\rm Log}(\theta)-(1-\xi-\theta){\rm Log}(1-\xi-\theta)-(z-1-\theta){\rm Log}(z-1-\theta)$$
and
$$H(n,n\xi)\sim \frac{1}{2^n}\int_0^\xi d\varphi\ e^{n S_1(\varphi,z)}\ ,$$
where
$$S_1(\varphi,z)=
-\varphi{\rm Log}(\varphi)-(1-\varphi){\rm Log}(1-\varphi)\ .$$
The extremum of the latter action is at $\varphi=\varphi^*=\frac{1}{2}$. Two cases must be distinguished:
\begin{itemize}
\item{(i)} $\xi>\frac{1}{2}$: then the asymptotics of $H(n,n\xi)$ are dominated by $e^{n S_1(\varphi^*,z)}/2^n\sim 1$
hence no contribution to the large $n$ dominant behavior.
\item{(ii)}  $\xi<\frac{1}{2}$: the integral is dominated by the value at its upper bound $\xi$, namely
$$H(n,n\xi)\sim e^{n S_1(\xi,z)-n{\rm Log}(2)}$$
\end{itemize}
Collecting both results, we find that $\frac{N(n,nz)}{Z(n,n)}$ is dominated by $e^{nS(\xi,\theta,z)}$, where:
\begin{itemize}
\item{(i)} If $\xi>\frac{1}{2}$:  $S(\xi,\theta,z)=S_0(\xi,\theta,z)$.
\item{(ii)} If $\xi<\frac{1}{2}$:  $S(\xi,\theta,z)=S_0(\xi,\theta,z)+S_1(\xi,z)-{\rm Log}(2)$.
\end{itemize}
Note that 
$$ \partial_\theta S_0(\xi,\theta,z)={\rm Log}\left( \frac{(1-\xi-\theta)(z-1-\theta)}{\theta(z-\xi-\theta)} \right) \ \ {\rm and} \ \ 
\partial_\xi S_0(\xi,\theta,z)={\rm Log}\left(  \frac{1-\xi-\theta}{z-\xi} \right) \ .$$
The extremization problem in the case (i), $\partial_\theta S_0(\xi,\theta,z)=\partial_\xi S_0(\xi,\theta,z)=0$ has clearly no solution
as one should have $\theta=1-z$, but $z>1$. We are left with (ii). Writing $\partial_\xi S=\partial_\theta S=0$, we get:
$$ (1-\xi-\theta)(z-1-\theta)=\theta(z-\xi-\theta) \ \ {\rm and}\ \ (1-\xi-\theta)(1-\xi)=(z-\xi-\theta)\xi \ .$$
Eliminating $\theta$ yields $\xi=\xi^*(z):=\frac{1}{2z}<\frac{1}{2}$ as $z>1$.

The line through the scaled points $(\xi^*(z),2-\xi^*(z))$ and $(z,z)$ has the equation:
\beq\label{paraline}
y=\frac{2-\xi^*(z)-z}{\xi^*(z)-z} x-2z \frac{1-\xi^*(z)}{\xi^*(z)-z} \ .
\eeq
Following Section \ref{tansec}, the arctic curve must be the envelope of this family of lines, hence it
is determined by the equation \eqref{paraline} together with its derivative w.r.t. $z$. 
This gives the system:
\begin{eqnarray*}
-x - y + 2 z + 4 x z - 4 z^2 - 2 x z^2 + 2 y z^2&=&0\\
1 + 2 x - y - 2 x y - 3 z - 2 x z + 2 x^2 z + 2 y z + 2 x y z&=&0
\end{eqnarray*}
whose solution is given the parametric equations, for $z>1$:
$$ x=\frac{1}{2}-\frac{( z-1) ( 2 z-1)}{2 z (z-1)+1},\quad y=\frac{3}{2}+\frac{z-1}{2 z (z-1)+1} \ .$$
Alternatively, eliminating $z$, we end up with the arctic circle:
\beq\label{arcticcircle}
x^2 +(y-1)^2=\frac{1}{2}
\eeq
centered at $(0,1)$ with radius $1/\sqrt{2}$, namely inscribed in the limiting scaled domain $\lim_{n\to\infty}D/n$ 
i.e. the square defined by the equation $|x|+|y-1|=1$.
{\it Stricto sensu},  we have only derived the portion of the arctic circle
corresponding to $z>1$, namely $-\frac{1}{2}<x<\frac{1}{2}$
and $\frac{3}{2}<y<2$.
However, due to the symmetry of the problem under rotation of $\pi/2$, the same analysis can be performed
on the rotated domains, leading to \eqref{arcticcircle} for all $x,y$.

\section{The non-intersecting Dyck path problem}\label{dycksecone}

In this section we study a different problem of non-intersecting lattice paths, and apply the tangent method
to determine the corresponding arctic curve.

\subsection{Counting problem}

\begin{figure}
\centering
\includegraphics[width=15.cm]{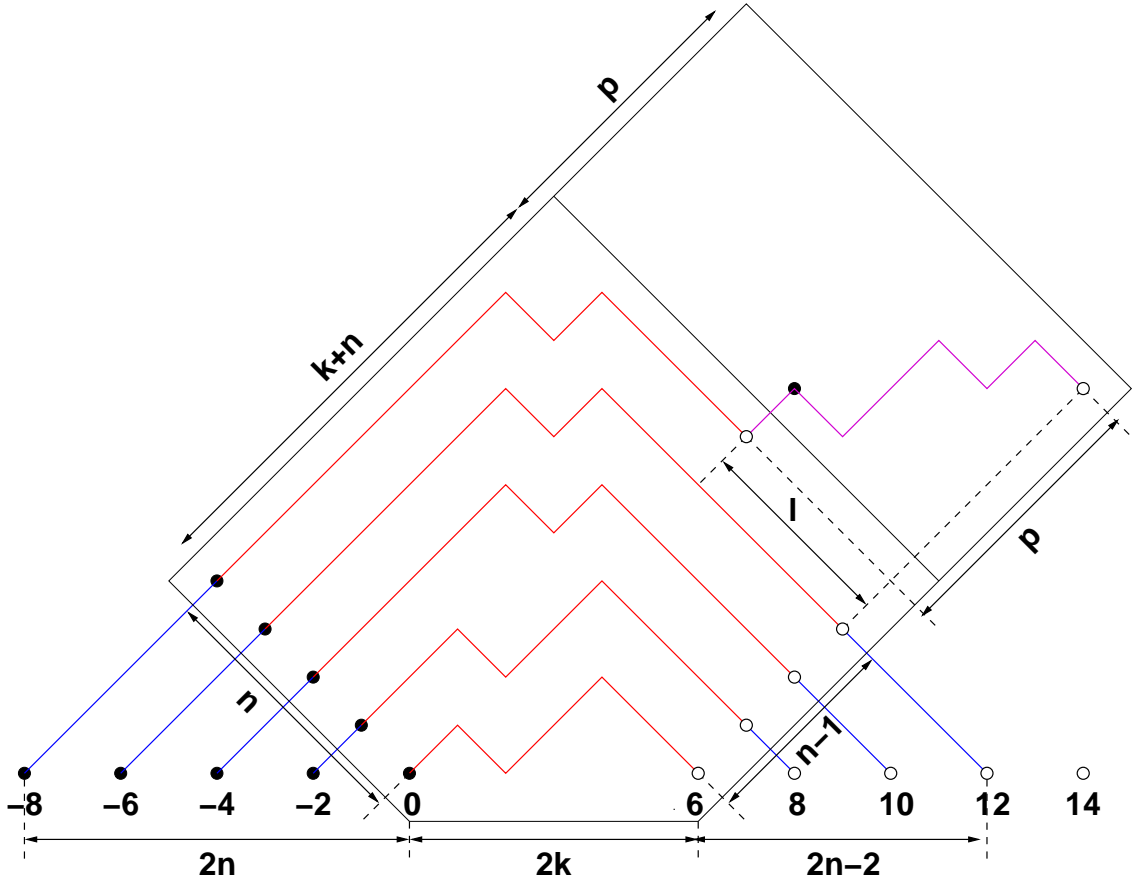}
\caption{\small A typical configuration contributing to $N(n,k,p)$, with $n=4$, $k=3$, $p=4$
and $\ell=3$. The rightmost path is split into three pieces: (i) its initial part inside the truncated $(n+k)\times (n+k)$ square,
(ii) one step $(1,1)$ exiting this square, (iii) its final part in the $p\times (n+k)$ rectangle.}
\label{fig:catatg}
\end{figure}

We consider the situation of Fig.~\ref{fig:catatg}. We wish to count the number $N(n,k,p)$
of path configurations with up/down steps
$(1,\pm 1)$ only, that remain above the x axis, and start at points $(-2i,0)$, $i=0,1,..,n$, and end at points
$(2k+2i,0)$, $i=0,1,...,n-1$ and $(2k+n+p,n+p)$.
Denote by $Z(n,k;\ell)$ the partition function corresponding to the lower $(n+k)\times (n+k)$ truncated square, and
$Y(p;\ell)$ that of the upper rectangle $p\times (n+k)$. Then we have
\beq
N(n,k,p)=\sum_{\ell=0}^{n+k} Z(n,k;\ell)\,Y(p;\ell)\ .
\eeq
We shall rather consider the ``one-point function"
\beq 
H(n,k;\ell):=\frac{Z(n,k;\ell)}{Z(n,k;0)}\ ,
\eeq
and we will apply the saddle-point method to the sum:
\begin{equation}\label{main}
\frac{N(n,k,p)}{Z(n,k;0)}=\sum_{\ell=0}^{n+k} H(n,k;\ell)\,Y(p;\ell)\ .
\end{equation}
The latter partition function $Y(p;\ell)$ is very simple, as it enumerates the configurations of a single path from 
$(2k+n+1-\ell,n+1+\ell)$ to
$(2k+n+p,n+p)$, namely
\begin{equation}\label{Yvalue}Y(p;\ell)={p+\ell-1\choose \ell} \ .
\end{equation}
The next two sections are devoted to the computation of $H(n,k;\ell)$.

\subsection{Partition function}

In this section, we compute the partition function $Z(n,k;0)$. Note that it may be obtained by specializing a result 
from 1989 by Gessel and Viennot (see Krattenthaler \cite{KrattCat} for details and proofs), 
however the method of proof is completely different. Here we derive  the result
using $LU$ decomposition, which will be instrumental for computing $Z(n,k;\ell)$ as well (see next section).

Note first that $Z(n,k;0)$ is simply the partition function of $n+1$  non-intersecting Dyck paths starting at points
$(-2i,0)$ and ending at points $(0,2k+2i)$, $i=0,1,...,n$, as the final portion of the left most path
is simply a straight line from $(2k+n,n)$ to $(2n+2k,0)$. 
By applying the Gessel-Viennot formula of Lemma~\ref{lemma:LGV}, we have:

\begin{thm}\label{enu1}
\beq Z(n,k;0)=\det_{i,j\in [0,n]}\left( c_{k+i+j} \right)= 
\prod_{i=0}^n U_{i,i}, \qquad U_{i,i}=\frac{(2i+1)!\,(2k+2i)!}{(k+2i+1)\, !(k+2i)!}
\eeq
where $c_m=\frac{1}{m+1}{2m\choose m}$ is the $m$-th Catalan number, enumerating the Dyck paths of $2m$ steps.
\end{thm}
\begin{proof}
By $LU$ decomposition of the matrix $A$ with entries $A_{i,j}=c_{k+i+j}$. Introduce the lower triangular matrix $L$ with entries:
\begin{eqnarray}L_{i,j}&=&2^{i-j} {i\choose j} \prod_{s=j+1}^{i} \frac{2k+2s-1}{k+1+j+s}\nonumber \\
&=&\frac{(2k+2i)!(k+j)!(k+2j+1)!}{(2k+2j)!(k+i)!(k+i+j+1)!}{i\choose j}=
\frac{{2k+2i\choose 2i-2j}{2i-2j\choose i-j}{i\choose j}}{{k+i\choose i-j}{k+1+i+j\choose i-j}}\ . \label{Lmatdyck}
\end{eqnarray}
Then $L^{-1}$ has entries:
\beq (L^{-1})_{i,j}=(-2)^{i-j} {i\choose j}\prod_{s=j+1}^{i} \frac{2k+2s-1}{k+i+s}=
(-1)^{i+j}\frac{(2k+2i)!(k+j)!(k+i+j)!}{(2k+2j)!(k+i)!(k+2i)!} {i\choose j}\ .
\eeq
Indeed, we have
\begin{eqnarray*}
&&\sum_{r=j}^i (L^{-1})_{i,r}L_{r,j}=\\
&&=\sum_{r=j}^i (-1)^{i+r}\frac{(2k+2i)!(k+r)!(k+i+r)!}{(2k+2r)!(k+i)!(k+2i)!}{i\choose r}
\frac{(2k+2r)!(k+j)!(k+2j+1)!}{(2k+2j)!(k+r)!(k+r+j+1)!}{r\choose j}\\
&&=\frac{(2k+2i)!(k+j)!(k+2j+1)!i!}{(2k+2j)!(k+i)!(k+2i)!j!}
\sum_{r=j}^i (-1)^{i+r}\frac{(k+i+r)!}{(k+j+r+1)!(i-r)!(r-j)!}\\
&&=\frac{1}{k+i+j+1}\frac{(2k+2i)!(k+j)!(k+2j+1)!i!}{(2k+2j)!(k+i)!(k+2i)!j!}
\sum_{r=j}^i (-1)^{i+r}{k+i+r\choose r-j}{k+i+j+1\choose i-r}\ .
\end{eqnarray*}
If $i=j$, the result is obviously $1$, as only the value $r=i=j$ contributes. If $i>j$, the sum reads:
\begin{eqnarray*}
&&(-1)^{i+j}\sum_{r=0}^{i-j} (-1)^{r}{k+i+j+r\choose r}{k+i+j+1\choose i-j-r}\\
&&\qquad =(-1)^{i+j}\sum_{r=0}^{i-j}
(1+x)^{k+i+j+1}\Big\vert_{x^{i-j-r}} \frac{1}{(1+x)^{k+i+j+1}}\Bigg\vert_{x^{r}}
=(-1)^{i+j}\frac{(1+x)^{k+i+j+1}}{(1+x)^{k+i+j+1}}\Bigg\vert_{x^{i-j}}=0
\end{eqnarray*}
by direct application of Lemma \ref{prodlem}.

Finally, let us compute for $i\geq j$:
\begin{eqnarray*}(L^{-1}A)_{i,j}&=&\sum_{r=0}^i (-1)^{i+r}\frac{(2k+2i)!(k+r)!(k+i+r)!}{(2k+2r)!(k+i)!(k+2i)!} {i\choose r} 
\frac{(2k+2r+2j)!}{(k+r+j+1)!(k+r+j)!}\\
&=&\frac{(2k+2i)!i!(i-j)!k!(2j)!}{(k+i)!(k+2i)!(k+i+j+1)!}\\
&&\qquad \times \sum_{r=0}^i(-1)^{r+i} {k+r+i\choose i-j}{k+i+j+1\choose i-r}{k+r\choose r}{2k+2r+2j\choose 2j}\ .
\end{eqnarray*}
We now use the following.
\begin{lemma}\label{helplem}
For all $i\geq j$, we have:
\beq
f_{k,i,j}:=\sum_{r=0}^i(-1)^{r+i} {k+r+i\choose i-j}{k+i+j+1\choose i-r}{k+r\choose r}{2k+2r+2j\choose 2j}=
\left\{ \begin{matrix} (2i+1){k+i\choose k} & {\rm if} \ i=j \\
0 & {\rm if} \ i>j 
\end{matrix}\right.\ .
\eeq
\end{lemma}
\begin{proof}
Let us write $j=i-m$. Then:
\begin{eqnarray*}
f_{k,i,i-m}&=&\sum_{r=0}^i(-1)^{r+i} {k+r+i\choose m}{k+r\choose k}{k+2i-m+1\choose i-r}{2k+2r+2i-2m\choose 2i-2m}\\
&=& \frac{1}{m!}\frac{d^m}{dt^m} \frac{1}{k!}\frac{d^k}{du^k} \sum_{r=0}^i t^{k+r+i} u^{k+r}
\left(\frac{(1-x)^{k+2i-m+1}}{(1-y)^{2i-2m+1}}\Bigg\vert_{x^{i-r}y^{2k+2r}}\right)\Bigg\vert_{t=u=1}\\
&=& \frac{1}{m!}\frac{d^m}{dt^m} \frac{1}{k!}\frac{d^k}{du^k}t^{k+2i} u^{k+i} \sum_{r=0}^i  (t u)^{r-i}
\left(\frac{(1-x^2)^{k+2i-m+1}}{(1-y)^{2i-2m+1}}\Bigg\vert_{x^{2i-2r}y^{2k+2r}}\right)\Bigg\vert_{t=u=1}\ ,
\end{eqnarray*}
where in the second line we have used Lemma \ref{iterdetlem} to represent ${k+r+i\choose m}$ and ${k+r\choose k}$
and Lemma \ref{fourwaylem} to represent the other two binomial coefficients,
and
where in the third line we performed a change of variable $x\to x^2$. 
Let us now use Lemma \ref{prodlem} to write:
$$\sum_{r=-k}^{i} (tu)^{r-i} \frac{(1-x^2)^{k+2i-m+1}}{(1-y)^{2i-2m+1}}\Bigg\vert_{x^{2i-2r}y^{2k+2r}}=
\frac{(1-\frac{x^2}{t u})^{k+2i-m+1}}{(1-x)^{2i-2m+1}}\Bigg\vert_{x^{2k+2i}}\ .$$
Note the bounds in the summation: they ensure that all the terms occurring in the product are
accounted for. Indeed, the power of $x^2$ should be between $0$ and $k+2i-m+1$, hence $0\leq i-r\leq k+2i-m+1$
or equivalently $m-i-1-k\leq r\leq i$, whereas the power of $y$ should be non-negative, i.e. $r\geq -k$. However $m=i-j\leq i$
hence the first lower bound is automatically satisfied. 
Finally, we note that
$$ \frac{1}{m!}\frac{d^m}{dt^m} \frac{1}{k!}\frac{d^k}{du^k}t^{k+2i} u^{k+i} \sum_{r=-k}^{-1} (t u)^{r-i}
\left(\frac{(1-x^2)^{k+2i-m+1}}{(1-y)^{2i-2m+1}}\Bigg\vert_{x^{2i-2r}y^{2k+2r}}\right)\Bigg\vert_{t=u=1}=0$$
as the derivatives w.r.t. $u$ vanish as soon as $r<0$ (the overall power of $u$ is $k+r<k$).
We may therefore rewrite
\begin{equation*}
f_{k,i,i-m}= \frac{1}{m!}\frac{d^m}{dt^m} \frac{1}{k!}\frac{d^k}{du^k}t^{k+2i} u^{k+i} 
\left(\frac{(1-\frac{x^2}{t u})^{k+2i-m+1}}{(1-x)^{2i-2m+1}}\Bigg\vert_{x^{2k+2i}}\right)\Bigg\vert_{t=u=1}\ .
\end{equation*}
It is clear that the quantity:
\begin{equation}\label{helpder}
\frac{1}{m!}\frac{d^m}{dt^m} \frac{1}{k!}\frac{d^k}{du^k} t^{k+2i} u^{k+i}\frac{(1-\frac{x^2}{t u})^{k+2i-m+1}}{(1-x)^{2i-2m+1}}\Bigg\vert_{t=u=1}
\end{equation}
is a polynomial of $x$, as the denominator after setting $u=t=1$ divides the numerator. The degree of this polynomial
is $2(k+2i-m+1)-(2i-2m+1)=2k+2i+1$, and we must pick the coefficient of $x^{2k+2i}$ in this polyniomial.
This coefficient is easily found by expanding Eq.~\eqref{helpder} around $x=\infty$:
$$t^{k+2i} u^{k+i}\frac{(1-\frac{x^2}{t u})^{k+2i-m+1}}{(1- x)^{2i-2m+1}}=
(-1)^{k+m}t^{k+2i} u^{k+i} \frac{x^{2k+2i+1}}{(t u)^{k+2i-m+1}}\left(1+\frac{2i-2m+1}{x}+O(x^{-2})\right)$$
and we conclude that
$$f_{k,i,i-m}=\frac{1}{m!}\frac{d^m}{dt^m} \frac{1}{k!}\frac{d^k}{du^k} (-1)^{k+m}(2i-2m+1)t^{m-1} u^{m-i-1}\vert_{t=u=1}\ .$$
For $m=0$ this gives
$$f_{k,i,i}=\frac{1}{k!}\frac{d^k}{du^k} (-1)^{k}(2i+1)u^{-i-1}\vert_{u=1}=(2i+1){k+i\choose k}\ ,$$
while for $m>0$ we find: $f_{k,i,i-m}=0$ as the derivative w.r.t. $t$ vanishes. The Lemma follows.
\end{proof}

We conclude that $L^{-1}A=:U$ is upper triangular, with diagonal elements:
\begin{eqnarray*}U_{i,i}&=&\frac{(2k+2i)!i!k!(2i)!}{(k+i)!(k+2i)!(k+2i+1)!}\times  (2i+1) \frac{(k+i)!}{i!k!}\nonumber \\
&=& \frac{(2i+1)!\,(2k+2i)!}{(k+2i+1)!\,(k+2i)!}
\end{eqnarray*}
and the Theorem follows.

\end{proof}

\subsection{One-point function}

We now turn to the computation of $H(n,k;\ell)$.

The partition function $Z(n,k;\ell)$ for general $\ell\in [0,n+k]$ corresponds by the Gessel-Viennot formula to the determinant of the 
matrix $C(n,k,\ell)$ with entries:
\beq C(n,k,\ell)_{i,j}=\left\{ \begin{matrix}
c_{k+i+j} & {\rm if}\ j\in [0,n-1]\\
b_{2k+2i+n-\ell,n+\ell}         & {\rm if}\ j=n 
\end{matrix}\right.\ .
\eeq
Indeed, as before $c_{k+i+j}$ is the number of Dyck paths from $(-2i,0)$ to $(2k+2j,0)$, while
$b_{2k+2i+n-\ell,n+\ell}$ is the number of Dyck paths from $(-2i,0)$ to $(2k+n-\ell,n+\ell)$,
where for $h=a$ mod 2, we have:
\beq 
b_{a,h}={a\choose \frac{a-h}{2}} -{a\choose \frac{a-h}{2}-1} =\frac{h+1}{\frac{a+h}{2}+1}{a\choose \frac{a-h}{2}}
\eeq
so that
\beq 
b_{2k+2i+n-\ell,n+\ell}=\frac{n+\ell+1}{n+i+k+1} \, {2k+2i+n-\ell\choose k+i-\ell}\ .
\eeq
Note also that for $\ell=0$, $b_{2k+2i+n,n}\neq c_{n+i+k}$, yet the determinant is the same as that with 
$\{c_{n+i+k}\}_{i\in [0,n]}$ as the last column. This is a simple manifestation of the fact that non-intersecting paths ending at
$(2k+2j,0)$, $j=0,1,...,n$ must end with a number of descending steps, which is $j$ for the $j$-th path counted from the bottom, 
hence in particular we may cut the topmost path to its last visited vertex before the $n$ forced descents, namely at the point
$(2k+n,n)$.

We have therefore 
\beq
Z(n,k;\ell)=\det_{i,j\in [0,n]} ({\tilde A}_{i,j}),\quad {\tilde A}_{i,j}=\begin{pmatrix} 
c_k & c_{k+1} & \cdots & c_{n+k-1} & b_{2k+n-\ell,n+\ell}\\
c_{k+1} & c_{k+2} & \cdots & c_{n+k} & b_{2k+n-\ell+2,n+\ell}\\
\vdots & \vdots &   & \vdots & \vdots \\
c_{n+k} & c_{n+k+1} & \cdots & c_{2n+k-1} & b_{2k+3n-\ell,n+\ell}
\end{pmatrix}\ .
\eeq

To compute this determinant, we use the result of Section \ref{sec:paths-and-LU}. 
Using the $L$ matrix of \eqref{Lmatdyck}, we get that $L^{-1}{\tilde A}={\tilde U}$
where ${\tilde U}$ is an upper triangular matrix differing from $U$ only in its last column,
with in particular:
\begin{eqnarray*}{\tilde U}_{n,n}&=&\sum_{r=0}^k (L^{-1})_{n,r} b_{2k+n-\ell+2r,n+\ell} \\
&=& \sum_{r=0}^n (-1)^{n+r}\frac{(2n+2k)!(k+r)!(n+k+r)!}{(2k+2r)!(n+k)!(2n+k)!} {n\choose r}
\frac{n+\ell+1}{n+r+k+1} \, {2k+2r+n-\ell\choose k+r-\ell}\ .
\end{eqnarray*}
As explained in Section \ref{sec:paths-and-LU}, the one-point function $H(n,k;\ell)$ is nothing but
the quantity 
\begin{eqnarray*}H(n,k;\ell)&=&\frac{{\tilde U}_{n,n}}{{U}_{n,n}}= (n+\ell+1)\frac{(k+2n+1)!n!}{(2n+1)!(n+k)!}\times \\
&&\qquad \times \sum_{r=0}^n (-1)^{n+r}\frac{(k+r)!(n+k+r)!(2k+2r+n-\ell)!}{(2k+2r)!r!(n-r)!(n+r+k+1)!(k+r-\ell)!}\ .
\end{eqnarray*}

We have the following.
\begin{thm}\label{onepthm}
The one-point function $H(n,k;\ell)=\frac{{\tilde U}_{n,n}}{{U}_{n,n}}$ reads:
\begin{eqnarray}\label{goodformula}
H(n,k;\ell)&=&\frac{(n+\ell+1)!(2k+n-\ell)!}{(2n+2k)!} 
\sum_{s=0}^n \frac{(2n+k-s)!}{(\ell+2s-n)!(n+k-\ell-s)!(2n+1-2s)!}\nonumber \\
&=&\frac{1}{{2n+2k\choose n+\ell}}\sum_{s=0}^n {n+\ell+1\choose 2n+1-2s}{2n+k-s\choose n+\ell}\ .
\end{eqnarray}
\end{thm}


In the remainder of this section the theorem is proved in two steps, as we must distinguish the cases $n\geq \ell$ and $n<\ell$.

\subsubsection{Case $n\geq \ell$}
We first note that for $n\geq \ell$, we may rewrite:
\begin{eqnarray*}{\tilde U}_{n,n}&=&(n+\ell+1)\frac{(2n+2k)!n!\ell!(n-\ell)!}{(k+2n)!(k+2n+1)!} \\
&&\qquad\qquad\quad \times \sum_{r=0}^n (-1)^{n+r}{n+k+r\choose n+k}{k+r\choose l}
{2k+2r+n-\ell\choose n-\ell}{k+2n+1\choose n-r}
\end{eqnarray*}

We have the following:
\begin{lemma}
For all $k\geq 0$ and all $n\geq \ell\geq 0$, we have:
\begin{eqnarray}g_{n,k,\ell}&:=&\sum_{r=0}^n (-1)^{n+r}{n+k+r\choose n+k}{k+r\choose l}
{2k+2r+n-\ell\choose n-\ell}{k+2n+1\choose n-r}\nonumber \\
&=&{2n+1\choose n-\ell}{n+\ell\choose n}
\end{eqnarray}
\end{lemma}
\begin{proof}
We compute:
\begin{eqnarray*}
g_{n,k,\ell}&=&\frac{1}{(n+k)!}\frac{d^{n+k}}{dt^{n+k}} \frac{1}{\ell!}\frac{d^\ell}{du^\ell} \sum_{r=0}^n
t^{n+k+r}u^{k+r}\left(\frac{(1-x)^{k+2n+1}}{(1-y)^{n-\ell+1}}\Bigg\vert_{x^{n-r}y^{2k+2r}}\right)\Bigg\vert_{t=u=1}\\
&=&\frac{1}{(n+k)!}\frac{d^{n+k}}{dt^{n+k}} \frac{1}{\ell!}\frac{d^\ell}{du^\ell}t^{k+2n}u^{n+k}\sum_{r=0}^n (t u)^{r-n}
\left(\frac{(1-x^2)^{k+2n+1}}{(1-y)^{n-\ell+1}}\Bigg\vert_{x^{2n-2r}y^{2k+2r}}\right)\Bigg\vert_{t=u=1}
\end{eqnarray*}
As before, we wish to use Lemma \ref{prodlem}, which in this case gives:
$$\sum_{r=-k}^n (t u)^{r-n}
\frac{(1-x^2)^{k+2n+1}}{(1-y)^{n-\ell+1}}\Bigg\vert_{x^{2n-2r}y^{2k+2r}}=\frac{(1-\frac{x^2}{t u})^{k+2n+1}}{(1-x)^{n-\ell+1}}\Bigg\vert_{x^{2k+2n}}$$
As before, the discrepancy between the summation ranges is resolved by noting that terms with $r<0$ contribute $0$ 
after taking $n+k$ derivatives w.r.t. $t$, as the corresponding overall power of $t$ is $n+k+r<n+k$.
We therefore have:
\begin{eqnarray*}
g_{n,k,\ell}&=&\frac{1}{(n+k)!}\frac{d^{n+k}}{dt^{n+k}} \frac{1}{\ell!}\frac{d^\ell}{du^\ell}t^{k+2n}u^{n+k}
\left( \frac{(1-\frac{x^2}{t u})^{k+2n+1}}{(1-x)^{n-\ell+1}}\Bigg\vert_{x^{2k+2n}}\right)\Bigg\vert_{t=u=1}\\
&=&\frac{1}{(n+k)!}\frac{d^{n+k}}{dt^{n+k}} \frac{1}{\ell!}\frac{d^\ell}{du^\ell}
t^{-k}u^{-n-k}\left(\frac{(1-t u x^2)^{k+2n+1}}{(1-t u x)^{n-\ell+1}}\Bigg\vert_{x^{2k+2n}}\right)
\Bigg\vert_{t=u=1}
\end{eqnarray*}
As before it is readily seen that the quantity
$$
\frac{1}{(n+k)!}\frac{d^{n+k}}{dt^{n+k}} \frac{1}{\ell!}\frac{d^\ell}{du^\ell}
t^{-k}u^{-n-k}\frac{(1-t u x^2)^{k+2n+1}}{(1-t u x)^{n-\ell+1}}
\Bigg\vert_{t=u=1}
$$
is a polynomial of $x$ as the denominator at $t=u=1$ always divides the numerator. The degree of this polynomial
is $2(k+2n+1)-(n-\ell+1)=2k+3n+\ell+1$. The desired result is obtained by an expansion around $x\to\infty$:
\begin{eqnarray*}&&t^{-k}u^{-n-k}\frac{(1-t u x^2)^{k+2n+1}}{(1-t u x)^{n-\ell+1}}\\
&&\qquad \qquad=(-1)^{n+k+\ell}\sum_{i=0}^{k+2n+1}\sum_{j\geq 0}t^{n+\ell-i-j}u^{\ell-i-j} {k+2n+1\choose i}{n-\ell+j\choose n-\ell}
x^{2k+3n+\ell+1-2i-j}
\end{eqnarray*}
and extracting the coefficient of $x^{2n+2k}$, namely imposing $2i+j=n+\ell+1$. This yields:
$$g_{n,k,\ell}=(-1)^{n+k+\ell}\frac{1}{(n+k)!}\frac{d^{n+k}}{dt^{n+k}} \frac{1}{\ell!}\frac{d^\ell}{du^\ell}\sum_{i=0}^{n} t^{i-1}u^{i-n-1}
{k+2n+1\choose i}{2n+1-2i\choose n-\ell}\Bigg\vert_{t=u=1}
$$
Note that each term with $i\geq 1$ yields zero after differentiation w.r.t. $t$, as $i-1<n<k+n$. We are left with the contribution of $i=0$,
which reads:
$$g_{n,k,\ell}=(-1)^{n+k+\ell}{2n+1\choose n-\ell}\frac{1}{(n+k)!}\frac{d^{n+k}}{dt^{n+k}} \frac{1}{\ell!}\frac{d^\ell}{du^\ell} t^{-1}u^{-n-1}\Bigg\vert_{t=u=1}=
{2n+1\choose n-\ell}{n+\ell\choose \ell}
$$
and the Lemma follows.
\end{proof}

%
%
%
%

We deduce that for $n\geq \ell$, we have:
$$\frac{{\tilde U}_{n,n}}{{U}_{n,n}}=
\frac{(n+\ell+1)\frac{(2n+2k)!n!\ell!(n-\ell)!}{(k+2n)!(k+2n+1)!}{2n+1\choose n-\ell}{n+\ell\choose \ell}}{\frac{(2n+1)!\,(2n+2k)!}{(k+2n+1)!\, (k+2n)!}}=1
$$ 
To identify this with the statement of Theorem \ref{onepthm} for $n\geq \ell$, we need to show that Eq.~\eqref{goodformula}  
gives $H(n,k;\ell)=1$ for $n\geq\ell\geq 0$. This is a consequence of the following:

\begin{lemma}\label{slem}
The following identity holds  for all $n\geq \ell\geq 0$:
$$ 
\sum_{s=0}^n {n+\ell+1\choose 2n+1-2s}{2n+k-s\choose n+\ell}={2n+2k\choose n+\ell}
$$
\end{lemma}
\begin{proof}
We write:
\begin{eqnarray*}
\sum_{s=0}^n {n+\ell+1\choose 2n+1-2s}{2n+k-s\choose n+\ell}
&=& \sum_{s=0}^n  (1+x)^{n+\ell+1}\Big\vert_{x^{\ell-n+2s}} 
\frac{1}{(1-y)^{n+\ell+1}}\Bigg\vert_{y^{n+k-\ell-s}}\\
&=&\sum_{s=0}^n  (1+x)^{n+\ell+1}\Big\vert_{x^{\ell-n+2s}} 
\frac{1}{(1-y^2)^{n+\ell+1}}\Bigg\vert_{y^{2(n+k-\ell-s)}}\\
&=&\left(\frac{1+x}{1-x^2}\right)^{n+\ell+1}\Bigg\vert_{x^{2k+n-\ell}}\\
&=&\frac{1}{(1-x)^{n+\ell+1}}\Bigg\vert_{x^{2k+n-\ell}}={2n+2k\choose n+\ell}
%
\end{eqnarray*}
where we have performed a change of variables $y\to y^2$ in the second line, and 
where to go from the second to the third line we noticed that the summation over $s$ can be relaxed to all 
$s\in[ 0,n+\ell+1]$, as the actual range of the contributing terms is for $s\in [\frac{n-\ell}{2},n]\subset [ 0,n+\ell+1]$, and applied Lemma \ref{prodlem}
(note that this would be wrong if $\ell>n$).
The result of the Lemma follows.
\end{proof}
This completes the proof of Theorem \ref{onepthm} in the case $n\geq \ell$.

\subsubsection{Case $\ell>n$}
We now consider the case $\ell>n$. 

We may write:
\begin{eqnarray}\frac{{\tilde U}_{n,n}}{U_{n,n}}
&=&(n+\ell+1)\frac{n!\ell!}{(2n+1)!} \nonumber \\
&&\quad \times \sum_{r=0}^n (-1)^{n+r}{n+k+r\choose n+k}{k+r\choose l}
{k+2n+1\choose n-r}\frac{(2k+2r+n-\ell)!}{(2k+2r)!}\label{ratioU}
\end{eqnarray}
In the following, we use the standard notation for Pochhammer symbols:
$$ (x)_m=x(x-1)\cdots (x-m+1)\ .$$
Introduce the following quantities, using respectively the expression of $\frac{{\tilde U}_{n,n}}{U_{n,n}}$ via 
Eq.~\eqref{ratioU} and
that of the sought after result for $H(n,k;\ell)$ via Eq.~\eqref{goodformula}:
\begin{eqnarray}
&&\qquad\qquad \qquad \qquad P_{n,\ell}(k):=\frac{(2n+1)!(n+\ell)! {2n+2k \choose n+\ell}}{(n+\ell+1)\, 
\ell! {n+k\choose \ell}}\frac{{\tilde U}_{n,n}}{U_{n,n}}\label{defP} \\
&=& \sum_{r=0}^n (-1)^{n-r}{n\choose r}\frac{(2k+2r+n-\ell)!}{(2k+n-\ell)!} \frac{(n+k+r)!}{(n+k)!}
\frac{(n+k-\ell)!}{(k+r-\ell)!} \frac{(k+2n+1)!}{(n+k+r+1)!}\frac{(2n+2k)!}{(n+k)!}\frac{(k+r)!}{(2k+2r)!}\nonumber \\
&=& \sum_{r=0}^n (-4)^{n-r} {n\choose r} (2k+2r+n-\ell)_{2r}\, (n+k+r)_r\,(n+k-\ell)_{n-r}\, (k + 2n + 1)_{n-r}\,  (n +k- \frac{1}{2} )_{n-r}\nonumber 
\end{eqnarray}
and
\begin{eqnarray}
Q_{n,\ell}(k)&:=&\frac{(2n+1)!(n+\ell)! {2n+2k \choose n+\ell}}{(n+\ell+1)\, \ell! {n+k\choose \ell}}H(n,k,\ell)\nonumber \\
&=&\frac{(2n+1)!}{n+\ell+1} \sum_{s=0}^n {n+\ell+1\choose 2n+1-2s}\, (k + 2n -s)_{n-s}\, (n +k-\ell)_{s}\label{defQ}
\end{eqnarray}
$P_{n,\ell}$ and $Q_{n,\ell}$ are two polynomials of $k$, with a priori degree $\deg(P_{n,\ell})=3n$ and $\deg(Q_{n,\ell})=n$, but we have the following:
\begin{thm}\label{pols}
The polynomials $P_{n,\ell}$ and $Q_{n,\ell}$ have same degree $n$. Moreover, they coincide on all the points $k=-j-n$, for
$j=1,2,...,n+1$, where they take the values:
\beq\label{values}
P_{n,\ell}(-j-n)=Q_{n,\ell}(-j-n)=(-1)^n \frac{(2 n + 1)!(j- 1)!(n+ 2 j + \ell  -1)!}{(n+\ell + 1)  (2 j - 1)!  (\ell + j - 1)! }
\eeq
\end{thm}
\begin{proof}
Let us first compute the degree of $P_{n,\ell}$. First note that the expression \eqref{defP} is valid for both $\ell>n$ and $\ell\leq n$.
But in the latter case, we have found that 
$$\frac{{\tilde U}_{n,n}}{U_{n,n}}=\frac{(n+\ell+1)\, \ell! {n+k\choose \ell}}{(2n+1)!(n+\ell)! {2n+2k \choose n+\ell}} P_{n,\ell}(k)=1$$
which implies that 
$$P_{n,\ell}(k)=\frac{(2n+1)!}{n+\ell+1} \, \frac{(2n+2k)_{n+\ell}}{(n+k)_\ell}$$ 
is a polynomial of degree $n$, as for $n\geq \ell$ the 
denominator always divides the numerator. We may infer that all the cancellations that reduce the degree from $3n$ to $n$ still occur
when $\ell>n$, and the degree of $P_{n,\ell}$ is always $n$. 
Let us now compute using \eqref{defP} for $k=-j-n$, $1\leq j\leq n+1$:
$$
P_{n,\ell}(-j-n)=\sum_{r=0}^n (-4)^{n-r} {n\choose r} (2r-2j-n-\ell)_{2r}\, (r-j)_r\,(-j-\ell)_{n-r}\, (n-j + 1)_{n-r}\,  (-j- \frac{1}{2} )_{n-r}
$$
Vanishing contributions to the sum arise from the terms 
$(r-j)_r$ if $r\geq j$ and $(n-j + 1)_{n-r}$ if $r<j-1$. We conclude that only the term $r=j-1$ survives, with value:
\begin{eqnarray*}
P_{n,\ell}(-j-n)&=&(-4)^{n-j+1} {n\choose j-1} (-2-n-\ell)_{2j-2}\\
&&\qquad \times  (-1)_{j-1}\,(-j-\ell)_{n-j+1}\, (n-j + 1)_{n-j+1}\,  (-j- \frac{1}{2} )_{n-j+1}\\
&=&(-2)^{n-j+1} {n\choose j-1}\frac{(2j+n+\ell-1)!}{(n+\ell+1)!}  \\
&& \qquad \times (-1)^{j-1} (j-1)!(-1)^{n-j+1}\frac{(n+\ell)!}{(j+\ell-1)!} (n-j+1)! (-1)^{n-j+1}\frac{(2n+1)!!}{(2j-1)!!}\\
&=&(-1)^n \frac{(2n+1)!(j-1)!(2j+n+\ell-1)!}{(n+\ell+1)(2j-1)!(j+\ell-1)!}
\end{eqnarray*}
which proves \eqref{values} for $P_{n,\ell}$. Finally, let us compute using \eqref{defQ} for $k=-j-n$, $1\leq j\leq n+1$:
\begin{eqnarray*}
Q_{n,\ell}(-j-n)&=&\frac{(2n+1)!}{n+\ell+1} \sum_{s=0}^n {n+\ell+1\choose 2n+1-2s}\, (n-j -s)_{n-s}\, (-j-\ell)_{s}
\end{eqnarray*}
The vanishing contributions to the sum correspond to $n-j \geq s$, hence we may rewrite:
\begin{eqnarray*}
Q_{n,\ell}(-j-n)&=&(-1)^n\frac{(2n+1)!}{n+\ell+1} \sum_{s=n-j+1}^n {n+\ell+1\choose 2n+1-2s}\,\frac{(j-1)!}{(j+s-n-1)!}\, \frac{(j+\ell+s-1)!}{(j+\ell-1)!}\\
&=&(-1)^n\frac{(2n+1)!(j-1)!}{(n+\ell+1)(j+\ell-1)!}\sum_{s=n-j+1}^n {n+\ell+1\choose 2n+1-2s}\,\frac{(j+\ell+s-1)!}{(j+s-n-1)!}\\
&=&(-1)^n\frac{(2n+1)!(j-1)!(2j+n+\ell-1)!}{(n+\ell+1)(2j-1)!(j+\ell-1)!}\\
&&\qquad \times \ 
\frac{1}{{2j+n+\ell-1\choose n+\ell}}\sum_{s=n-j+1}^n {n+\ell+1\choose 2n+1-2s}\, {j+\ell+s-1\choose n+\ell}
\end{eqnarray*}
To conclude, we need a variant of Lemma \ref{slem}, proved in the same way:
\begin{lemma}\label{tlem}
We have the identity:
$$\frac{1}{{2j+n+\ell-1\choose n+\ell}}\sum_{s=n-j+1}^n {n+\ell+1\choose 2n+1-2s}\, {j+\ell+s-1\choose n+\ell}=1$$
for $\ell>n$ and $1\leq j\leq n+1$.
\end{lemma}
This is proved by computing:
\begin{eqnarray*}
\sum_{s=n-j+1}^n {n+\ell+1\choose 2n+1-2s}\, {j+\ell+s-1\choose n+\ell}&=&\sum_{s=n-j+1}^n (1+x)^{n+\ell+1}\vert_{x^{2n+1-2s}}
\frac{1}{(1-x^2)^{n+\ell+1}}\Bigg\vert_{x^{2(j+s-n-1)}} \\
&=& \frac{1}{(1-x)^{n+\ell+1}}\Bigg\vert_{x^{2j-1}}={2j+n+\ell-1\choose 2j-1}
\end{eqnarray*}
by application of Lemma \ref{prodlem}.
We deduce that \eqref{values} holds for $Q_{n,\ell}$ as well, and Theorem \ref{pols} follows.
\end{proof}
This completes the proof of Theorem \ref{onepthm} in the case $\ell>n$.




\subsection{Tangent method and arctic curve}

\begin{figure}
\centering
\includegraphics[width=12.cm]{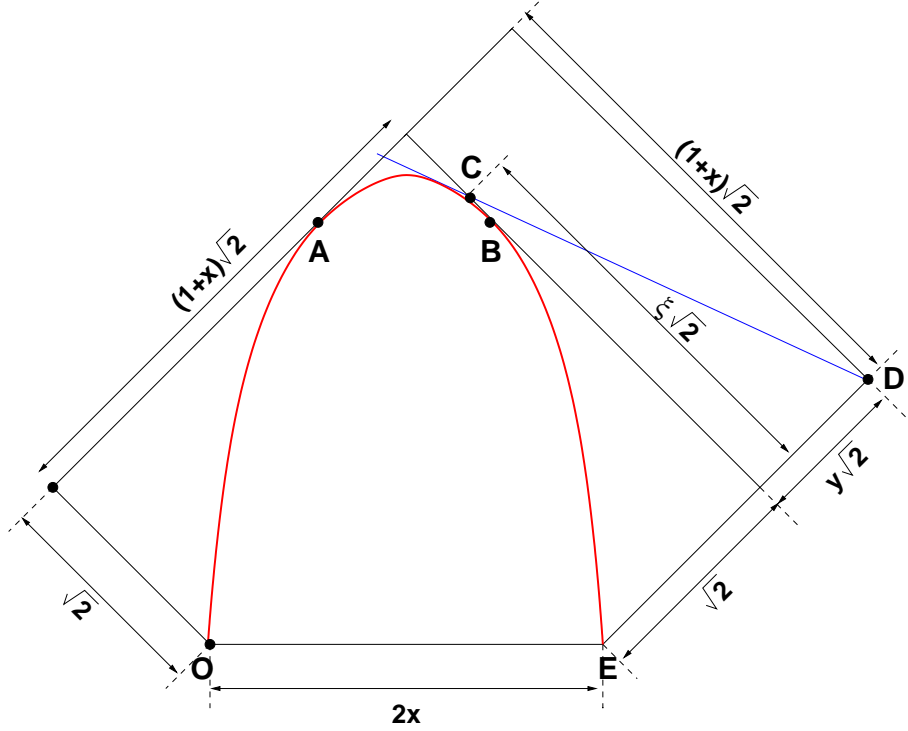}
\caption{\small The scaling of the path domain in units of $n\to\infty$, with $k=x n$, $p=y n$, $\ell=\xi n$,
and the arctic ellipse. The points $O,A,B,C,D$ have respective coordinates
$(0,0)$, $\big(\frac{2x}{2+x},\frac{4(1+x)}{2+x}\big)$, $\big(\frac{2x(1+x)}{2+x},\frac{4(1+x)}{2+x}\big)$, $(2x+1-\xi,1+\xi),(2x+1+y,1+y)$. 
The line $CD$ is tangent to the arctic curve.}
\label{fig:contg}
\end{figure}

Using the results of previous section, we are now ready to apply the saddle-point approximation to the normalized sum:
\begin{equation}\label{total}
\frac{N(n,k,p)}{Z(n,k;0)}=\sum_{\ell=0}^{n+k} H(n,k;\ell)\, Y(p;\ell)
\end{equation}

The tangent method described in Section \ref{tansec}, 
applied to the present problem, is illustrated in Fig.~\ref{fig:contg}. It consists of finding which value of 
$\ell$ dominates the contribution to the sum
\eqref{total}, by the saddle-point method. The corresponding
parametric family of lines $CD$ is tangent to the arctic curve, which is uniquely defined as their envelope.

We wish to evaluate the quantity $N(n,k,p)/Z(n,k;0)$ at large $n,k,p$ via the saddle-point method. Write
$k=x n$, $p=y n$, $n$ large. The summation in \eqref{total} is approximated by the integral
$$\frac{N(n,x n,y n)}{Z(n,x n;0)}\simeq n \int_{0}^{1+x} d\xi \, H(n,xn;\xi n) \, Y(yn;\xi n) $$
by writing $\ell=\xi n$ and $\sum_{\ell=0}^{n+k} \to n\int_{0}^{1+x} d\xi $.
For large $n$, using the expression \eqref{Yvalue},  we have by the Stirling formula:
\begin{equation}\label{Yas}
Y(y n;\xi n)\simeq {n(y+\xi)\choose n\xi } \simeq e^{n S_0(y,\xi)}
\end{equation}
where
\begin{equation}\label{Szero} S_0(y,\xi)=(y+\xi){\rm Log}(y+\xi)-\xi {\rm Log}(\xi) -y {\rm Log}(y) 
\end{equation}

Similarly, we write the summation \eqref{goodformula} as:
\begin{eqnarray*}
H(n,x n;\xi n)&\simeq&  \frac{(n+\xi n)!(2x n+ n-\xi n)!}{(2n+2x n)!} 
\, \\
&&\qquad\qquad \times\ n \int_{0}^1 d\sigma \frac{(2n+x n-\sigma n)!}{(\xi n+2\sigma n-n)!(n+x n-\xi n-\sigma n)!(2n-2\sigma n)!}\\
&\simeq& \int_{0}^1 d\sigma\, e^{S_1(x,\xi,\sigma)}
\end{eqnarray*}
where we have set $s=\sigma n$ and replaced $\sum_{s=0}^n\to n\int_0^1 d\sigma$, and:
\begin{eqnarray} S_1(x,\xi,\sigma)&=&(1+\xi){\rm Log}(1+\xi)+(2x+1-\xi) {\rm Log}(2x+1-\xi) 
-(2+2x) {\rm Log}(2+2x) \nonumber \\
&&+(2+x-\sigma) {\rm Log}(2+x-\sigma) -(\xi+2\sigma-1) {\rm Log}(\xi+2\sigma-1)\nonumber \\
&&-(1+x-\xi-\sigma){\rm Log}(1+x-\xi-\sigma)-(2-2\sigma){\rm Log}(2-2\sigma)\label{Sone}
\end{eqnarray}

Assembling all the pieces, we are left with saddle-point equations for the action
$S(x,y,\xi,\sigma)=S_0(y,\xi)+S_1(x,\xi,\sigma)$, whereas the variables $\xi,\sigma$
must satisfy:
\begin{equation}\label{domain}
0\leq \sigma \leq 1, \quad  0\leq \xi \leq 1+x, \quad \xi+\sigma\leq 1+x,\quad 
\xi+2\sigma \geq 1
\end{equation}
We find:
\begin{eqnarray*}
\partial_\xi S=0&=& {\rm Log}\left(\frac{(y+\xi)(1+\xi)(1+x-\xi-\sigma)}{\xi(2x+1-\xi)(\xi+2\sigma-1)}\right)\\
\partial_\sigma S=0&=&  
{\rm Log}\left(\frac{4(1-\sigma)^2(1+x-\xi-\sigma)}{(\xi+2\sigma-1)^2 (2+x-\sigma)}\right)
\end{eqnarray*}
The second equation boils down to:
$$\sigma=\frac{(2+3x)-\xi(2+x)}{3+4x-\xi} $$
We see that all conditions \eqref{domain} are automatically satisfied (as $x,y>0$) 
except for the bound $\sigma\geq 0$.
Define
$$\xi^*:=\frac{2+3x}{2+x}=1+\frac{2}{2+x}\, x\in (1,x+1)$$

If $\xi<\xi^*$: then $\sigma>0$, and the first equation gives:
$$y (2x + 1 - \xi)^2 (1 + \xi)=0$$
which has no physical solution. Hence we must have:
$\xi>\xi^*$, but then the second equation cannot hold, which means that the variable $\sigma$ is saturated
to be $\sigma=0$. In that case, we end up with the first equation at $\sigma=0$, namely a cubic
relation for $\xi$:
$$A(x,y,\xi):=(y+\xi)(1+\xi)(1+x-\xi)-\xi(2x+1-\xi)(\xi-1)=0$$
with a unique solution $\xi(x,y)$ in the physical domain \eqref{domain}.

We note  that
$A(x,0,\xi)=\xi((2+3x)-(2+x)\xi)$, with the only admissible solution $\xi=\xi^*$.
We deduce that the arctic curve is tangent to the right boundary line 
(with equation $v=2+2x-u$ in the $(u,v)$ plane) at the point
$$B=(2x+1-\xi^*,1+\xi^*)=\left(\frac{2x(1+x)}{2+x},\frac{4(1+x)}{2+x}\right) $$
Similarly,
for large $y$ we have $\lim_{y\to\infty} y^{-1}A(x,y,\xi)=(1+\xi)(1+x-\xi)$, with only admissible solution
$\xi=1+x$. We deduce that the arctic curve is tangent to the left boundary 
(with equation $v=2x+u$ in the $(u,v)$ plane) at some point $A$.
$A$ and $B$ are represented in Fig. \ref{fig:contg} in scale of $n$.

Assuming the tangent to the sought-after arctic curve is the parametric line $(CD)$:
$v=a(x,y)u +b(x,y)$, with $\xi=\xi(x,y)$, we have
$$1+y=(2x+1+y)a+b,\qquad 1+\xi=(2x+1-\xi)a+b$$
hence
$$a=\frac{y-\xi(x,y)}{y+\xi(x,y)},\qquad b=2\frac{(1+x+y)\xi(x,y)-y}{y+\xi(x,y)}$$
Note that $A(x,y,\xi)=0$ may be inverted as
$$y=y(x,\xi):=\frac{ \xi ((2+x)\xi-(2+3x ))}{(1+\xi)(1+x-\xi)}
$$
The envelope of the parametric family of lines 
$F(u,v,\xi,x):=v-a(x,y(\xi))u-b(x,y(\xi))$ is the solution of the system
$$F(u,v,\xi,x)=0\quad {\rm and}\quad \partial_\xi F(u,v,\xi,x)=0 $$
Eliminating $\xi$ between these equations, we find the ellipse:
\begin{equation}\label{arctic}
x^2v^2-4(1+x)u(2x-u)=0
\end{equation}
Strictly speaking, we have only computed the portion of arctic curve between the two tangency points $A$ and $B$,
namely for $u\in [\frac{2x}{2+x},\frac{2x(1+x)}{2+x}]$.
The entire half-ellipse with $v\geq 0$ turns out to be the full arctic curve for our problem, as shown in next section.

\section{An equivalent rhombus tiling problem}\label{dycksectwo}

In this section we present an equivalent rhombus tiling problem to that of non-intersecting Dyck paths
of the previous section, and use it to derive the remaining portions 
of the arctic ellipse \eqref{arctic}.

\subsection{The equivalent tiling problem}

\begin{figure}
\centering
\includegraphics[width=13.cm]{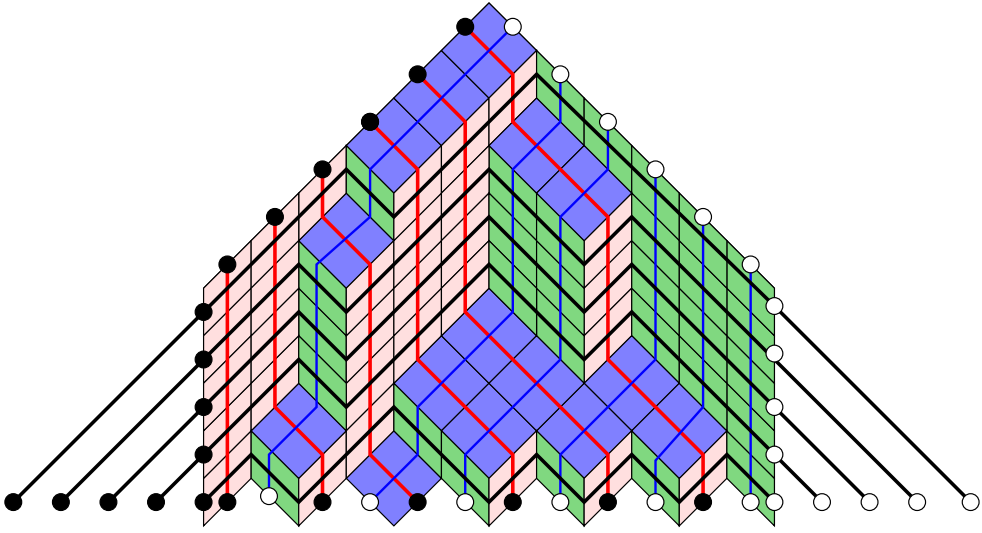}
\caption{\small The rhombus tiling problem corresponding to the non-intersecting Dyck paths of Sections 1 and 2, here with $n=4$ and $k=6$.
We have represented the Dyck paths (black) as well as two other families of paths (red and blue), each determining 
the tiling completely.}
\label{fig:tiling}
\end{figure}

The problem studied so far has an interpretation in terms of rhombus tilings of a ``half-hexagon", as illustrated in Fig. \ref{fig:tiling}. 
Such a tiling is entirely determined by either of three families of non-intersecting paths, which 
connect two parallel edges in each tile (the three families of paths are represented in black, blue and red respectively).

We are now in position to complete the previous study of the arctic curve for the Dyck path problem. 
We simply have to apply the tangent method for the other two families of paths. We concentrate on the red paths,
as the treatment of the blue paths is completely identical, up to reflection w.r.t. to the vertical line through $(k,0)$.

\subsection{Exact enumeration}

\begin{figure}
\centering
\includegraphics[width=13.cm]{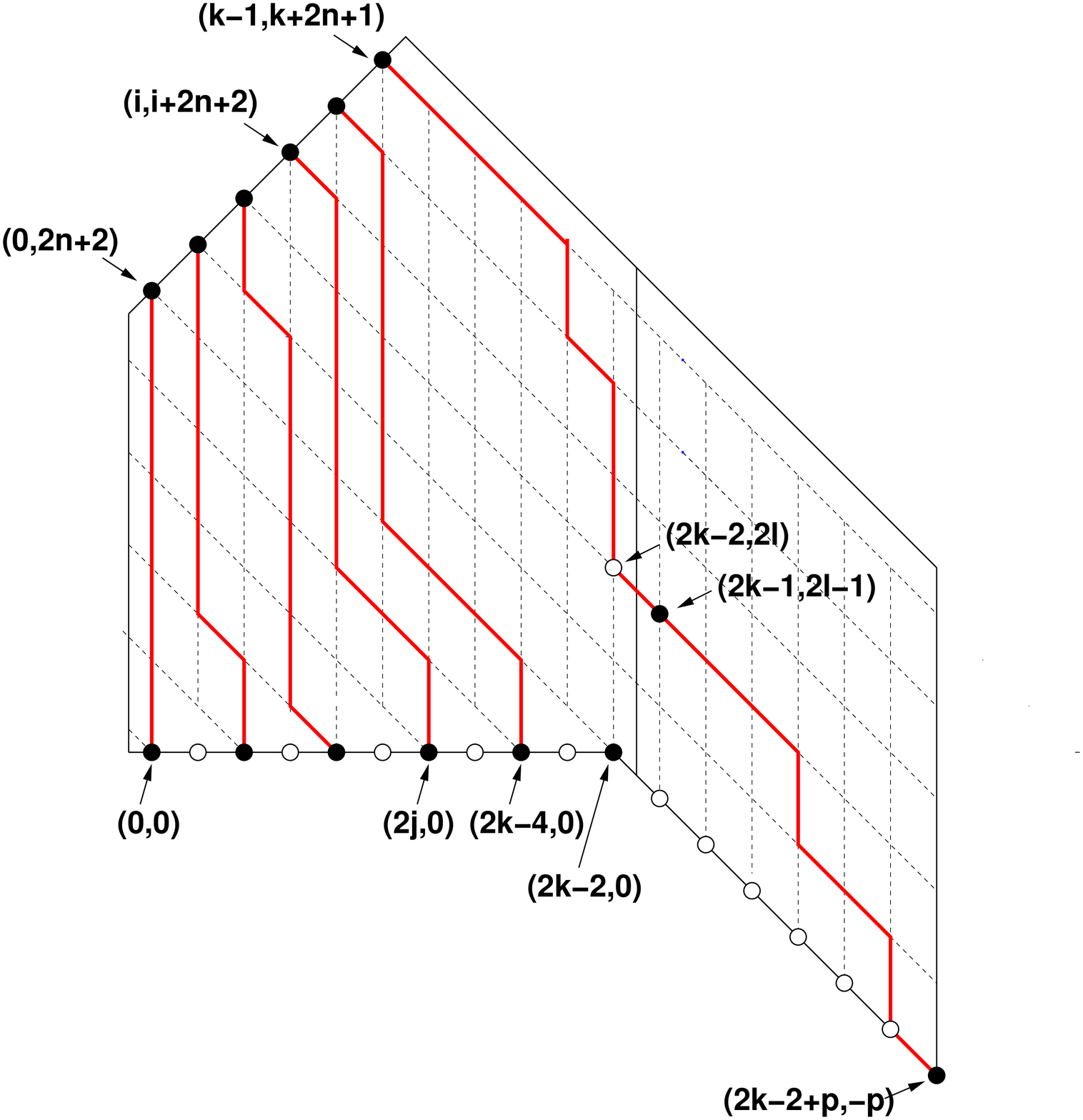}
\caption{\small The tangent method for the red paths of Fig. \ref{fig:tiling}.}
\label{fig:tilingtwo}
\end{figure}

We may now consider the red path family depicted in Fig. \ref{fig:tilingtwo}
that has starting points at $(i,i+2n+2)$, $i=0,1,...,k-1$ and endpoints at
$(2j,0)$, $j=0,1,...,k-2$ as well as the remote point $(2k-2+p,-p)$.
Note that the vertical steps are $(0,-2)$ and the diagonal ones are $(1,-1)$.
Let us denote by $N_2(n,k,p)$ the total number of path configurations in this family.
As before, we introduce the numbers $Z_2(n,k;\ell)$ and $Y_2(p;\ell)$ that count respectively the non-intersecting
family from $(i,i+2n+2)$, $i=0,1,...,k-1$ to $(2j,0)$, $j=0,1,...,k-2$ as well as $(2k-2,2\ell)$, and the configurations of
a single path from $(2k-1,2\ell-1)$ to $(2k-2+p,-p)$.
We have the decomposition formula:
\begin{equation}\label{maintwo}
N_2(n,k,p)=\sum_{\ell=0}^{n+1} Z_2(n,k;\ell)\,Y_2(p;\ell)
\end{equation}

The latter number is easily found to be:
$$Y_2(p;\ell)={p+\ell-1\choose \ell} =Y(p;\ell)$$
as the path has $p-1$ diagonal and $\ell$ vertical steps.

The number $Z_2(n,k;0)$ is given via Gessel-Viennot by the following determinant:
$$
Z_2(n,k;0)=\det(A), \quad A_{i,j}={j+n+1 \choose 2j-i},\quad  i,j=0,1,...,k-1
$$

\begin{thm}\label{enu2}
The number $Z_2(n,k;0)$ is given by:
$$Z_2(n,k;0)=\prod_{i=0}^{k-1} U_{i,i},\qquad U_{i,i}=\frac{(2n+2+2i)!i!}{(2n+2+i)!(2i)!}$$
\end{thm}
\begin{proof}
We proceed as before by $LU$ decomposition of $A$.
The $LU$ decomposition of $A$ uses the lower triangular matrix $L$ with elements:
$$ L_{i,j}=\frac{{2i-2j\choose i-j} {i \choose 2(i-j)}}{{i+2n+2\choose i-j}}=\frac{i! (j+2n+2)!}{(i-j)!(2j-i)!(i+2n+2)!}$$
and the entries of its inverse read:
$$ (L^{-1})_{i,j}=(-1)^{i+j} \frac{(j+2n+2)!(i-1)!}{(i+2n+2)!(j-1)!} {2i-j-1\choose i-j} =(-1)^{i+j}\frac{(j+2n+2)!(2i-j-1)!}{(i+2n+2)!(i-j)!(j-1)!} $$
This is readily checked by computing for $i\geq j$:
\begin{eqnarray*}
\sum_{m=0}^{k-1} (L^{-1})_{i,m}\,L_{m,j}&=&\sum_{m=0}^{k-1} (-1)^{i+m}\frac{(m+2n+2)!(2i-m-1)!}{(i+2n+2)!(i-m)!(m-1)!}
\frac{m! (j+2n+2)!}{(m-j)!(2j-m)!(m+2n+2)!}\\
&=&(-1)^i\frac{(j+2n+2)!}{(i+2n+2)!}\sum_{m=j}^{{\rm Min}(i,2j)} m(-1)^m\, \frac{(2i-m-1)!}{(2j-m)!(i-m)!(m-j)!}
\end{eqnarray*}
If $i=j$, this is clearly equal to $1$ as only $m=i=j$ contributes. If $i>j$, we have:
\begin{eqnarray*}
&&\!\!\!\!\!\!\!\!\!\!\!\!\!\!\!\!\!\!\!\!\!\!\!\!\!\!\!\!\!\!\!\!\!\!\!\!\!\!\!\!\!\!\!\!\!\!\!\!\!\!\!\!\!\!\!\!
 \frac{(i-j)!}{(2i-2j-1)!}\sum_{m=j}^{{\rm Min}(i,2j)} m(-1)^m\, \frac{(2i-m-1)!}{(2j-m)!(i-m)!(m-j)!}\\
&=&
\sum_{m=j}^{{\rm Min}(i,2j)} m(-1)^m\, {2i-m-1\choose 2i-2j-1}{i-j\choose m-j} \\
&=& (-1)^j \sum_{m=j}^{{\rm Min}(i,2j)} m\, x^{-2j}(1-x)^{-(2i-2j)}y^j(1-y)^{i-j}\Bigg\vert_{x^{-m}y^{m}}\\
&=& (-1)^j \sum_{m=j}^{{\rm Min}(i,2j)} x^{-2j}(1-x)^{-(2i-2j)}y\frac{d}{dy}\left(y^j(1-y)^{i-j}\right)\Bigg\vert_{x^{-m}y^{m}}\\
&=&(-1)^j y^{2j}(1-\frac{1}{y})^{-(2i-2j)}y\frac{d}{dy}\left(y^j(1-y)^{i-j}\right)\Bigg\vert_{y^0}\\
&=&
(-1)^j y^{2i}\frac{1}{(1-y)^{2i-2j}}y\frac{d}{dy}\left(y^j(1-y)^{i-j}\right)\Bigg\vert_{y^0}=0
\end{eqnarray*}
as the series of $y$ has $y^{2i+j}$ in factor, and hence has no constant term. Note that we have used Lemma \ref{prodlem}
to go from the fourth to the fifth line above.

We have $L^{-1}A=U$, where $U$ is upper triangular.
To see this we compute for $i>j$:
\begin{eqnarray*}
&\sum_{m=0}^{k-1} &(L^{-1})_{i,m} {A}_{m,j}=\sum_{m=0}^{k-1} (-1)^{i+m}\frac{(m+2n+2)!(2i-m-1)!}{(i+2n+2)!(i-m)!(m-1)!}{j+n+1 \choose 2j-m}\\
&=& \frac{((n+1+j)!)^2 (i-1)!}{(i+2n+2)!(2j-1)!}
\sum_{m=0}^{k-1} (-1)^{i+m}{2j-1\choose m-1}{2i-m-1\choose i-m}{m+2n+2\choose m+n+1-j}\\
&=& \frac{((n+1+j)!)^2 (i-1)!}{(i+2n+2)!(2j-1)!}
\sum_{m=0}^{k-1} \frac{(1+x)^{2j-1}}{x^{2j}} \frac{1}{y^i(1+y)^i}\frac{1}{z^{n+1-j}(1-z)^{j+n+2}}\Bigg\vert_{x^{-m}y^{-m}z^{m}}\\
&=& \frac{((n+1+j)!)^2 (i-1)!}{(i+2n+2)!(2j-1)!}
\frac{1-\left(\frac{x y}{z}\right)^{k}}{1-\frac{x y}{z}}\frac{(1+x)^{2j-1}}{x^{2j}} 
\frac{1}{y^i(1+y)^i}\frac{1}{z^{n+1-j}(1-z)^{j+n+2}}\Bigg\vert_{x^{0}y^{0}z^{0}}\\
&=& \frac{((n+1+j)!)^2 (i-1)!}{(i+2n+2)!(2j-1)!}\oint \frac{1}{(2i\pi)^3}\frac{dx dy dz}{xyz} \frac{1}{1-\frac{x y}{z}} \frac{(1+x)^{2j-1}}{x^{2j}} 
\frac{1}{y^i(1+y)^i}\frac{1}{z^{n+1-j}(1-z)^{j+n+2}}
\end{eqnarray*}
where we noted that the contribution of the term $(xy/z)^{k}$ in the fourth line vanishes as the corresponding 
$y$ residue vanishes (as $k-i>0$). Expressing the residue at $x=0$ in terms of that at $x=z/y$, we get:
\begin{eqnarray*}
&\sum_{m=0}^{k-1}  &(L^{-1})_{i,m} {A}_{m,j}=\frac{((n+1+j)!)^2 (i-1)!}{(i+2n+2)!(2j-1)!}
\oint \frac{1}{(2i\pi)^2}\frac{dy dz}{yz}  \frac{(1+z/y)^{2j-1}}{(z/y)^{2j}} 
\frac{1}{y^i(1+y)^i}\frac{1}{z^{n+1-j}(1-z)^{j+n+2}}\\
&=&\frac{((n+1+j)!)^2 (i-1)!}{(i+2n+2)!(2j-1)!}\oint \frac{dy dz}{(2i\pi)^2} \frac{(y+z)^{2j-1}}{ y^i(1+y)^i z^{j+n+2}(1-z)^{j+n+2}}\\
\end{eqnarray*}
Let us write the above as $R_{y\to 0}$
to indicate that it picks up the residue at $y=0$. 
First note that by expressing the residue at $z=0$ in terms of that at $z=1$, we get the identity:
\begin{equation}\label{otherR}
R_{y\to 0}=\frac{((n+1+j)!)^2 (i-1)!}{(i+2n+2)!(2j-1)!}\oint \frac{dy dz}{(2i\pi)^2} 
\frac{(y+1-z)^{2j-1}}{ y^i(1+y)^i z^{j+n+2}(1-z)^{j+n+2}}
\end{equation}
Next, let us write:
$R_{y\to 0}=-R_{y\to -1}-R_{y\to\infty}$ by use of the Cauchy theorem. 
To compute $R_{y\to -1}$ we perform the change of variables $y=-1-t$ in the original integral, leading to:
\begin{eqnarray*}R_{y\to -1}&=&-\frac{((n+1+j)!)^2 (i-1)!}{(i+2n+2)!(2j-1)!}\oint \frac{dt dz}{(2i\pi)^2} 
\frac{(-t-1+z)^{2j-1}}{t^i(1+t)^i z^{j+n+2}(1-z)^{j+n+2}}\\
&=&\frac{((n+1+j)!)^2 (i-1)!}{(i+2n+2)!(2j-1)!}\oint \frac{dt dz}{(2i\pi)^2}
\frac{(t+1-z)^{2j-1}}{t^i(1+t)^i z^{j+n+2}(1-z)^{j+n+2}}=R_{y\to 0}
\end{eqnarray*}
by comparing with the second expression \eqref{otherR} above. 
We deduce that $R_{y\to 0}=-\frac{1}{2}R_{y\to\infty}$, with 
\begin{eqnarray*}
R_{y\to \infty}&=&-\frac{((n+1+j)!)^2 (i-1)!}{(i+2n+2)!(2j-1)!}\oint \frac{1}{(2i\pi)^2}\frac{du dz}{u  z^{j+n+2}} u^{2i-2j}
\frac{(1+u z)^{2j-1}}{(1+u)^i (1-z)^{j+n+2}}\\
\end{eqnarray*}
where we have performed the change of variables $y=1/u$ to express $R_{y\to\infty}$ as $R_{u\to 0}$.
If $i>j$, the residue at $u=0$ vanishes, henceforth $(L^{-1}A)_{i,j}=-\frac{1}{2}R_{y\to \infty}=U_{i,j}=0$, 
and therefore $U$ is upper triangular.
Finally, if $i=j$, we find:
\begin{eqnarray*}
U_{i,i}&=&\frac{1}{2}\frac{((n+1+i)!)^2 (i-1)!}{(i+2n+2)!(2i-1)!}\oint \frac{1}{2i\pi}\frac{dz}{z^{i+n+2}} 
\frac{1}{(1-z)^{i+n+2}}\\
&=&\frac{((n+1+i)!)^2 i!}{(i+2n+2)!(2i)!}{2i+2n+2\choose i+n+1}=\frac{(2n+2+2i)!i!}{(2n+2+i)!(2i)!}
\end{eqnarray*}
and the Theorem follows.
\end{proof}

\begin{remark}
Comparing Theorems \ref{enu1} and \ref{enu2}, we obtain the following identity:
$$\prod_{i=0}^n \frac{(2i+1)!\,(2k+2i)!}{(k+2i+1)\, !(k+2i)!}=\prod_{i=0}^{k-1} \frac{(2n+2+2i)!i!}{(2n+2+i)!(2i)!}$$
\end{remark}

Note in particular that we have:
\begin{equation}\label{unn}
U_{k-1,k-1}=\frac{(2k+2n)!(k-1)!}{(2n+k+1)!(2k-2)!}
\end{equation}

The number $Z_2(n,k;\ell)$ is given via the Gessel-Viennot formula by the following determinant:
$$
Z_2(n,k;\ell)=\det({\tilde A}),\quad {\tilde A}_{i,j}:=\left\{ 
\begin{matrix} 
{j+n+1 \choose 2j-i} & {\rm if}\ j=0,1,...,k-2\\
{} & {} \\
{n+k-\ell\choose 2k-2-i} & {\rm if} \ j=k-1
\end{matrix}\right.
$$

As before, let us introduce the one-point function:
$H_2(n,k;\ell):=Z_2(n,k;\ell)/Z_2(n,k;0)$. Defining the upper triangular matrix ${\tilde U}$ 
such that ${\tilde A}=L{\tilde U}$, we may rewrite following Section \ref{sec:paths-and-LU}:
$$H_2(n,k;\ell)=\frac{{\tilde U}_{k-1,k-1}}{U_{k-1,k-1}}$$

We have the following.
\begin{thm}
The one-point function $H_2(n,k;\ell)$ reads:
\begin{equation}
H_2(n,k;\ell)=
\frac{2}{{2n+2k\choose 2n+3}} \sum_{s=\ell}^{n+1}  {k+n-s\choose k-2}{k+n+s-1\choose k-2}
\label{Hexp}
\end{equation}
\end{thm}
\begin{proof}
We have:
\begin{eqnarray*}{\tilde U}_{k-1,k-1}&=&\sum_{i=0}^{k-1} (L^{-1})_{k-1,i} {\tilde A}_{i,k-1}=
\sum_{i=0}^{k-1} (-1)^{i+k-1}\frac{(i+2n+2)!(2k-i-3)!}{(k+2n+1)!(k-1-i)!(i-1)!} {n+k-\ell\choose 2k-2-i}\\
&=&\frac{(2n+3)!(k-2)!}{(k+2n+1)!}\sum_{i=0}^{k-1} (-1)^{i+k-1}{2n+2+i\choose i-1} {2k-i-3\choose k-1-i}
{n+k-\ell \choose 2k-2-i} \\
\end{eqnarray*}
which implies:
\begin{eqnarray*}\frac{{\tilde U}_{k-1,k-1}}{U_{k-1,k-1}}&=&
\frac{2}{{2n+2k\choose 2n+3}}\sum_{i=0}^{k-1} (-1)^{i+k-1}{2n+2+i\choose i-1} {2k-i-3\choose k-1-i}
{n+k-\ell \choose 2k-2-i}\\
&=&\frac{2}{{2n+2k\choose 2n+3}}\sum_{i=k-1}^{n+k-\ell} (-1)^{i}{2n+2k-i\choose 2n+3} {i-1\choose k-2}
{n+k-\ell \choose i} \\
&=&\frac{2}{{2n+2k\choose 2n+3}}\oint \frac{1}{(2i\pi)^2} \frac{dx }{x^{2k-2}}\frac{dy}{y^{n-\ell+2}} \frac{(x+y)^{n+k-\ell}}{(1-x)^{2n+4}(1+y)^{k-1}}
\end{eqnarray*}
Let us rewrite the latter integral by expressing the residue at $x=0$ in terms of that at $x=1$, by the change of variables $x=1-u$:
\begin{eqnarray*}\frac{{\tilde U}_{k-1,k-1}}{U_{k-1,k-1}}&=&
\frac{2}{{2n+2k\choose 2n+3}}\oint \frac{1}{(2i\pi)^2} \frac{du }{u^{2n+4}}\frac{dy}{y^{n-\ell+2}} \frac{(1-u+y)^{n+k-\ell}}{(1-u)^{2k-2}(1+y)^{k-1}}
\end{eqnarray*}
Let us now change variables to $y=\frac{(1-u) v}{1-v}$:
\begin{eqnarray*}\frac{{\tilde U}_{k-1,k-1}}{U_{k-1,k-1}}&=&
\frac{2}{{2n+2k\choose 2n+3}}\oint \frac{1}{(2i\pi)^2} \frac{du }{u^{2n+4}}\frac{dv}{v^{n-\ell+2}} \frac{1}{(1-v)(1-u)^{k-1}(1-u v)^{k-1}}\\
&=&\frac{2}{{2n+2k\choose 2n+3}}  \frac{1}{(1-v)(1-u)^{k-1}(1-u v)^{k-1}}\Bigg\vert_{u^{2n+3}v^{n+1-\ell}}\\
&=&\frac{2}{{2n+2k\choose 2n+3}} \sum_{i=0}^{n+1-\ell} \frac{v^i}{(1-u)^{k-1}(1-u v)^{k-1}}\Bigg\vert_{u^{2n+3}v^{n+1-\ell}}\\
&=&\frac{2}{{2n+2k\choose 2n+3}} \sum_{i=0}^{n+1-\ell} {k+n-\ell-i-1\choose k-2}{n+k+\ell+i \choose k-2}\\
&=&\frac{2}{{2n+2k\choose 2n+3}} \sum_{i=0}^{n+1-\ell} {k+i-2\choose k-2}{k+2n+1-i \choose k-2}
\end{eqnarray*}
which completes the proof of \eqref{Hexp} up to the change of summation variable $s=n+1-i$, and the theorem follows.
\end{proof}

\subsection{Tangent method}

Taking $n\to\infty$ with $k=x n$, $\ell=\xi n$, $s=\sigma n$ in \eqref{Hexp}, we find:
$$
H_2(n,x n;\xi n)\simeq n \,\int_{\xi}^1 d\sigma\ e^{n S_1'(x,\sigma)}
$$
where 
\begin{eqnarray*}
S_1'(x,\sigma)&=&2 {\rm Log}(2) +2 x {\rm Log}(2 x)-(1+\sigma){\rm Log}(1+\sigma)-(1-\sigma){\rm Log}(1-\sigma)-2 x{\rm Log}(x)\\
&& -(2+2x){\rm Log}(2+2x)+(1+x+\sigma){\rm Log}(1+x+\sigma)+(1+x-\sigma){\rm Log}(1+x-\sigma)
\end{eqnarray*}

Similarly, we may use the asymptotic expression for $Y_2(p;\ell)=Y(p;\ell)$ for $p=y n$ of Eqs. (\ref{Yas}-\ref{Szero}).
The integration variables must belong to the domain:
$0\leq \xi \leq \sigma\leq x$
We now have to perform the saddle-point method for the total action $S(x,y,\xi,\sigma)=S_0(y,\xi)+S_1'(x,\sigma)$
and the integral:
$$\int_{0}^x d\xi \int_{\xi}^1 d\sigma\ e^{n\, S(x,y,\xi,\sigma)}$$

We first write:
$$\partial_\sigma S=0={\rm Log}\left(\frac{(1-\sigma)(1+x+\sigma)}{(1+\sigma)(1+x-\sigma)}\right)$$
with the unique solution $ \sigma=0$. This solution is never selected, and we have the saturation by the lower
bound $\sigma=\xi$. This results in the action
$S'(y,x,\xi)=S_0(y,\xi)+S_1'(x,\xi)$, with the bound $0\leq \xi\leq 1$.
We then write the saddle-point condition:
$$\partial_\xi S'=0={\rm Log}\left( \frac{(y+\xi)}{\xi} \frac{(1-\xi)(1+x+\xi)}{(1+\xi)(1+x-\xi)}\right)
$$

We find that for $0\leq \xi \leq x$, we have
$$y=y(\xi):=\frac{2x \xi^2}{(1 - \xi) (1 + x + \xi)} \in (0,\infty)$$

Assuming the tangent to the sought-after arctic curve is the parametric line through $(2x+y,-y)$
and $(2x,2\xi)$, with equation in the $(u,v)$ plane:
$v=a(x,y)u +b(x,y)$, with $y=y(\xi)$, and
$$-y=(2x+y)a+b,\qquad 2 \xi=2x a+b$$
hence
$$a(\xi)=-\frac{y(\xi)+2\xi}{y(\xi)},\qquad b(\xi)=2\frac{(\xi y(\xi)+2x\xi+xy(\xi))}{y(\xi)}$$
As before, the envelope of the family of lines $F(u,v,\xi)=v-a(\xi)u-b(\xi)$ is the solution fo
the system $F=0$ and $\partial_\xi F=0$.
We find that $ \xi=\frac{v x}{2u}$,
and finally, as $\xi\leq 1$:
$$x^2v^2-4(1+x)u(2x-u)=0, \qquad u\in \Big[\frac{2x(1+x)}{2+x},2\Big]$$
which is nothing but the portion of the arctic ellipse \eqref{arctic} corresponding to 
$u\in \Big[\frac{2x(1+x)}{2+x},2\Big]$.

Reflection of the problem w.r.t. the vertical line $u=1$ shows that the remaining portion of the arctic curve is another
portion of the {\it same} ellipse, corresponding to $u\in \Big[0,\frac{2x}{x+2}\Big]$.

\section{Another path model}\label{everyothersec}

\begin{figure}
\centering
\includegraphics[width=15.cm]{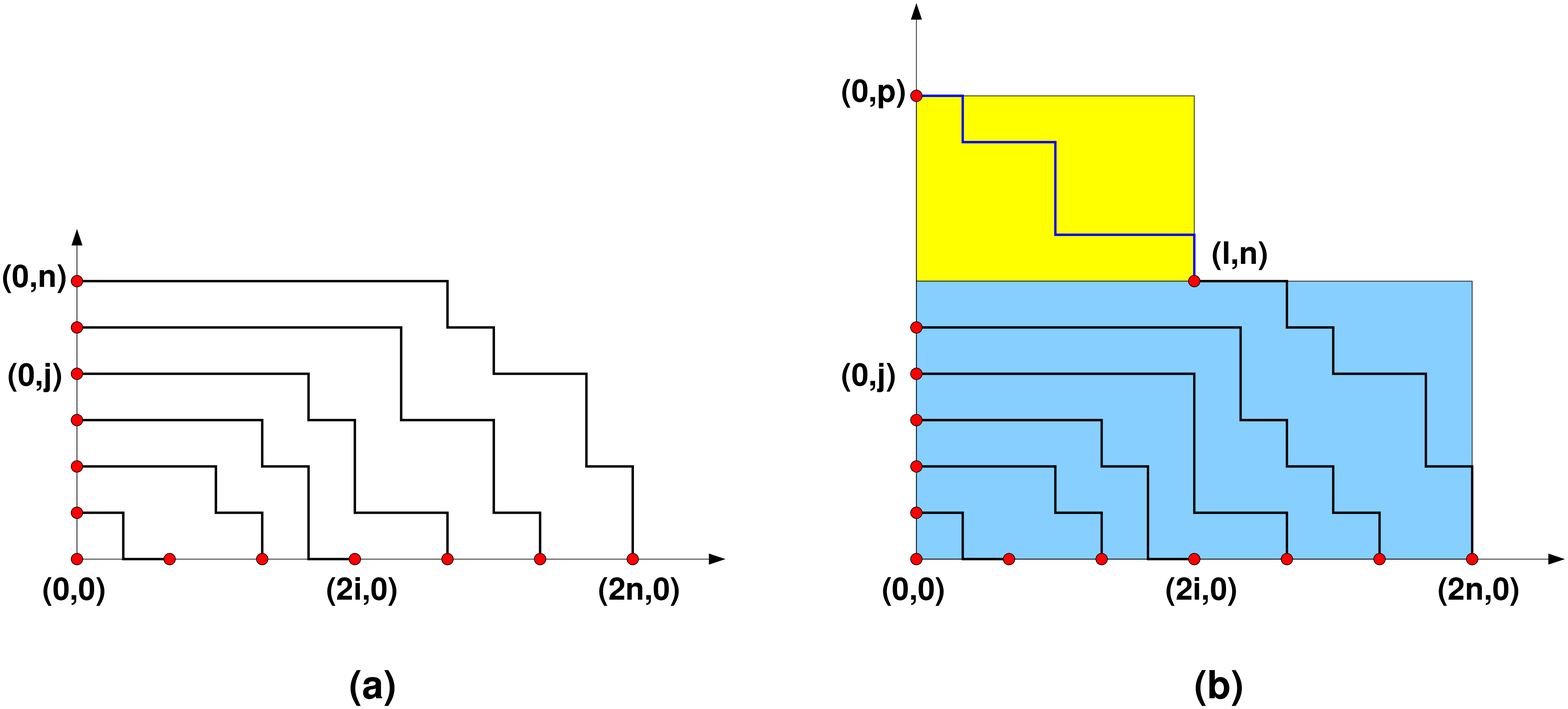}
\caption{\small (a) A typical configuration of the non-intersecting square-lattice path model with steps 
$(-1,0)$, $(0,1)$, from the points $(2i,0)$,  $i=0,1,...,n$ to the points $(0,j)$, $j=0,1,...,n$. (b) The setting for 
applying the tangent method: the topmost path escapes to reach a point $(0,p)$ on the vertical axis.}
\label{fig:everyother}
\end{figure}

We consider non-intersecting square-lattice paths in the first quadrant, 
with left and up steps $(-1,0)$, $(0,1)$,  that start at the points $(2i,0)$, 
$i=0,1,...,n$ and end up at the points $(0,j)$, $j=0,1,...,n$, as illustrated in Fig.~\ref{fig:everyother}(a).

\subsection{Partition function}

The partition function $Z(n)$ for the above paths is given by the Gessel-Viennot formula.
\begin{thm}
\beq\label{partother}
Z(n)=\det_{i,j\in [0,n]}\left( {2i+j\choose j} \right) =2^{n(n+1)/2}
\eeq
\end{thm}
\begin{proof}
Let $A$ be the matrix with entries $A_{i,j}={2i+j\choose j}$, $i,j=0,1,...,n$.
We prove the result by $LU$ decomposition. Consider the lower triangular matrix $L$ with entries $L_{i,j}={i\choose j}$
and its inverse $L^{-1}$ with entries $(L^{-1})_{i,j}=(-1)^{i+j} {i\choose j}$. Then we have the following result
$$(L^{-1} A)_{i,j}=U, \quad{\rm where}\quad  U_{i,j}=0\ {\rm for} \ i>j\ {\rm and}\ U_{i,i}=2^i$$
Let us indeed compute:
\begin{eqnarray*}
(L^{-1} A)_{i,j}&=&\sum_{k=0}^i (-1)^{i+k} {i\choose k} {2k+j\choose j}= \oint \frac{dt }{2i\pi t}
\left(\frac{1}{t^2}-1\right)^i \frac{1}{(1-t)^{j+1}}\\
&=&\oint \frac{dt }{2i\pi t^{2i+1}} (1+t)^i (1-t)^{i-j-1}
\end{eqnarray*}
where the last contour integral picks up the coefficient of $t^{2i}$. If $0\leq j<i$ the integrand 
is a polynomial of degree $\leq 2i-1$ hence the result is $0$. If $j=i$, the residue at the pole $t=1$ 
gives the result $2^i$.
We conclude that $Z(n)=\prod_{i=0}^n U_{i,i}= 2^{n(n+1)/2}$ and the theorem follows.
\end{proof}

\begin{remark}
Note that the partition for this path model \eqref{partother} coincides with the number of domino tilings of the
Aztec diamond \eqref{aztecpart}. There does not seem to exist any obvious natural bijection between the two sets of 
corresponding non-intersecting paths,
and we leave this as an open question for the reader. As an indication of the difficulty of the question, we shall find that the arctic curve
for the present problem is a portion of parabola, as opposed to a circle for the Aztec diamond case.
\end{remark}

\subsection{Partition functions for an escaping path}

As before, we consider now the situation where the topmost endpoint $(0,n)$ is replaced by a
sliding point $(0,p)$ with $p\geq n$ (see Fig.~\ref{fig:everyother}(b)). The total partition function, which
we denote by $N(n,p)$, then splits as
 $$ N(n,p)= \sum_{\ell=0}^{2n}Z(n;\ell)W(n,p;\ell)\ ,$$
where 
$Z(n;\ell)$ is the partition for the same paths (bottom rectangle in Fig.~\ref{fig:everyother}(b)), 
but endpoints $(0,j)$, $j=0,1,...,n-1$ and $(\ell,n)$, $\ell=0,1,2,...,2n$,
and
$W(n,p;\ell)$ is the partition function for a single path (top rectangle in Fig.~\ref{fig:everyother}(b))
from $(\ell,n)$ to $(0,p)$, with its first step being up
(to indicate that the point $(\ell,n)$ is the point when the top path exits the rectangle $[0,2n]\times [0,n]$).
The latter is easily computed to be
$$W(n,p;\ell)={p-n-1+\ell\choose \ell}$$

Let us now compute the one-point function:
$$H(n;\ell) := \frac{Z(n;\ell)}{Z(n)}=\frac{1}{2^{n(n+1)/2} }  \det_{i,j\in [0,n]} ({\tilde A}),\quad {\tilde A}_{i,j}=\left\{ 
\begin{matrix}
A_{i,j} & {\rm if}\ 0\leq j<n \\
& \\
{n+2i-\ell\choose n} & {\rm if}\ j=n
\end{matrix}\right. $$

We have the following.

\begin{thm}
The one-point function $H(n;\ell)$ reads:
$$H(n;\ell)=\frac{{\tilde U}_{n,n}}{U_{n,n}}= \frac{1}{2^n} \sum_{k=0}^{{\rm Min}(n,2n-\ell)} {n\choose k}$$
\end{thm}
\begin{proof}
We note again that 
$L^{-1}{\tilde A}={\tilde U}$ is still upper triangular, and that
\begin{eqnarray*}
{\tilde U}_{n,n}&=&\sum_{k} (L^{-1})_{n,k} {\tilde A}_{k,n}=\sum_{k=0}^n (-1)^{n+k} {n\choose k} {n+2k-\ell\choose n}\\
&=& \sum_{k=0}^n \left(\frac{1}{t^2}-1\right)^n\Bigg\vert_{t^{-2k}}   \frac{t^\ell}{(1-t)^{n+1}}\bigg\vert_{t^{2k}}
=\left(\frac{1-t^2}{t^2}\right)^n  \frac{t^\ell}{(1-t)^{n+1}}\Bigg\vert_{t^0}\\
&=& \frac{(1+t)^n}{1-t}\Bigg\vert_{t^{2n-\ell}}=\sum_{k=0}^{{\rm Min}(n,2n-\ell)} {n\choose k}
\end{eqnarray*}
where in the second line we have used Lemma \ref{prodlem}.
The theorem follows.
\end{proof}

\subsection{Tangent method}

Let us evaluate the function $N(n,p)=\sum_{\ell=0}^{2n} H(n;\ell)\, W(n,p;\ell)$ in the scaling limit where
$n$ is large and $p=z n$, $\ell=\xi n$, $k=x n$.
We have
$$N(n,p) \sim \int_{0}^2 d\xi \int_{0}^{{\rm Min}(1,2-\xi)}dx\ e^{n S(z,x,\xi)}$$
where $S(z,x,\xi)=S_0(x,\xi)+S_{1}(z,\xi)$, and
$$H(n;\xi n)\sim \int_{0}^{{\rm Min}(1,2-\xi)}dx\ e^{-n(x\, {\rm Log}(x)+(1-x){\rm Log}(1-x))}\sim \int_{0}^{{\rm Min}(1,2-\xi)}dx\  e^{n S_0(x,\xi)}$$
while
$$W(n,z n;\xi n)\sim e^{n((z-1+\xi){\rm Log}(z-1+\xi)-(z-1){\rm Log}(z-1)-\xi {\rm Log}(\xi))}\sim e^{n S_1(z,\xi)}$$
The first saddle point equation, w.r.t. $x$  reads as follows. The extremum of $S_0(x,\xi)$
is reached at $x=1/2$ only if $2-\xi>1/2$, i.e. $\xi<3/2$, in which case the critical value of the action is $S_0^*={\rm Log}(2)$.
When $\xi>3/2$, the extremum is at $x=2-\xi$, and the critical value of the action is
$S_0^*=(2-\xi){\rm Log}(2-\xi)+(\xi-1){\rm Log}(\xi-1)$.
The second saddle-point equation, w.r.t. $\xi$ reads as follows.
If $\xi <3/2$ we have the equation:
$$\partial_\xi S_1(z,\xi)=0 \quad \leftrightarrow\ z-1+\xi=\xi $$
which has no solution for $z>1$.
Therefore we have only solutions when $\xi>3/2$, that obey the equation:
$$\partial_\xi ( S_0^*+S_1)=0 \quad \leftrightarrow\ (z-1+\xi)(\xi-1)=\xi(2-\xi) $$
This gives the relation
$$ 2 \xi^2 +(z-4)\xi -(z-1)=0 \ \Rightarrow \ z=z(\xi):=\frac{(1-\xi)^2}{2-\xi}$$
The line through the points $(\ell,n)$ and $(0,p)$ becomes in the scaling limit the line
through  the points $(\xi,1)$ and $(0,z)$, with equation
$$y+\frac{z(\xi)-1}{\xi} \, x -z(\xi) =0\ \Rightarrow \ 2 + 3 x - 2 y - (4+ 2 x-y) \xi + 2 \xi^2=0$$
The envelope is obtained by eliminating $\xi$ between the above and its derivative w.r.t. $\xi$:
$$6 + x - 8\xi + 2 \xi^2=0$$
This finally yields the parabola:
\beq \label{parabo}
-8 x + 4 x^2 + 8 y - 4 x y + y^2=0\ .
\eeq
As before we have derived only the portion of this curve corresponding to $3/2<\xi<2$, i.e. $3/2<x<2$. However, similar 
considerations as above show that the entire range $1<\xi<2$, i.e. $0<x<2$ is valid, as detailed in the next sections.

\subsection{Tiling formulation and enumeration}

\begin{figure}
\centering
\includegraphics[width=16.cm]{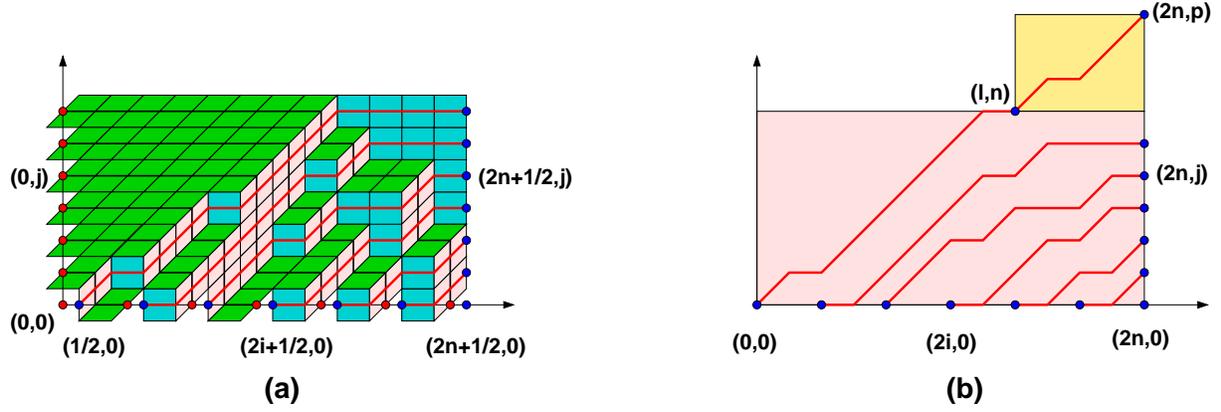}
\caption{\small (a) The tiling interpretation of the original path model, together with an alternative path description,
with steps $(1,0)$ and $(1,1)$, from the points $(2i+1/2,0)$, $i=0,1,...,n$ to
the points $(2n+1/2,j)$, $j=0,1,...,n$. (b) The setting for applying the tangent method: the topmost path escapes to reach a point $(2n,p)$ on the vertical axis (note that all coordinates have been shifted by $(-1/2,0)$ for simplicity).}
\label{fig:everyothertile}
\end{figure}

The path model of Fig.~\ref{fig:everyother}(a) has an alternative interpretation as a rhombus tiling 
problem, as illustrated in Fig.~\ref{fig:everyothertile}(a), where the path edges cross two of the three types of tiles transversally. This gives rise to an alternative description in terms of paths across two different types of tiles.
These paths have steps $(1,0)$ and $(1,1)$ and go from the points $(2i+1/2,0)$, $i=0,1,...,n$ to
the points $(2n+1/2,j)$, $j=0,1,...,n$. Shifting all coordinates by $(-1/2,0)$, these are paths with the abovementioned steps from the points $(2i,0)$, $i=0,1,...,n$ to
the points $(2n,j)$, $j=0,1,...,n$. Note that this is not equivalent to a reflection of the original path problem, as the steps
are different.

The partition function $Z_2(n)$ for this model is given by the Gessel-Viennot formula.
\begin{thm}
The partition function $Z_2(n)$ reads:
$$Z_2(n)=\det_{i,j\in [0,n]}\left({2(n-i)\choose n-j}\right)=\det_{i,j\in [0,n]}\left({2i\choose j}\right)=2^{n(n+1)/2}$$
\end{thm}
\begin{proof}
This is again proved by $LU$ decomposition. Define the matrix $A$ with entries $A_{i,j}={2i \choose j}$.
Then the usual lower triangular matrix $L$ with entries $L_{i,j}={i\choose j}$ is such that $L^{-1}A=U$ is upper triangular,
with diagonal elements $U_{i,i}=2^i$, and the theorem follows.
\end{proof}

\subsection{Tangent method}

We now apply the tangent method, by allowing the topmost path to end at some arbitrary point $(n,p)$, thus
splitting the corresponding partition function $N_2(n,p)$ into two terms, according to the position $(\ell,n)$ of the escape point 
from  the rectangle $[0,2n]\times[0,n]$, where $\ell\in[0,2n]$. We have
$$N_2(n,p)=\sum_{\ell=0}^{2n} Z_2(n;\ell) W_2(n,p;\ell) $$
where $Z_2(n;\ell)$ is the partition function for the same paths, but with the endpoint $(2n,n)$ moved to $(\ell,n)$,
namely:
$$Z_{n,\ell}=\det({\tilde A}), {\tilde A}_{i,j} =\left\{ \begin{matrix} 
A_{i,j} & {\rm if} \ 0\leq j<n\\
& \\
{\ell+2i-2n \choose n} & {\rm if}\  j=n
\end{matrix}\right. $$
and $W_2(n,p;\ell)$ is the partition function of a single path from $(\ell,n)$ to $(2n,p)$, with first step $(1,1)$ to guarantee
the escape from the rectangle $[0,2n]\times [0,n]$. We have
$$W_2(n,p;\ell)={2n-\ell-1\choose p-n-1}\ .$$

\begin{figure}
\centering
\includegraphics[width=9.cm]{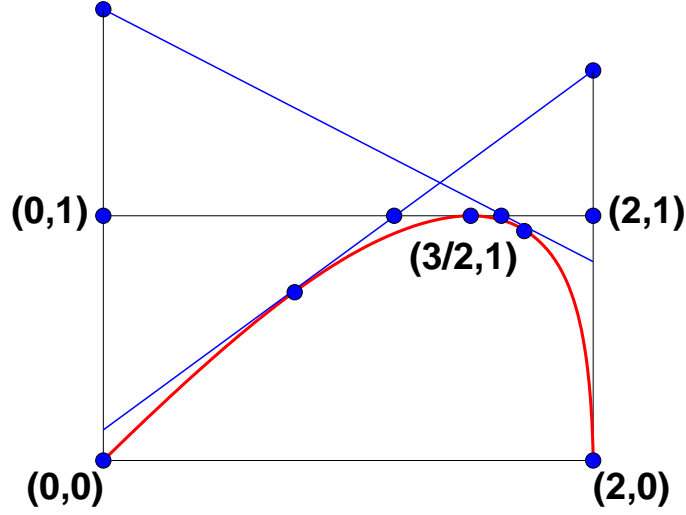}
\caption{\small The arctic parabola in the continuum limit separates frozen phases (top left and right corners)
from a disordered phase (bottom). We have represented two sample tangents obtained from the methods above,
together with their intercepts on the line $y=1$.}
\label{fig:everyothercont}
\end{figure}

As before we will apply the scaling analysis to the ratio $N_2(n,p)/Z_2(n)=\sum_\ell H_2(n;\ell) W_2(n,p;\ell)$
where the one-point function is defined as: 
$$H_2(n;\ell)=\frac{Z_2(n;\ell)}{Z_2(n)}$$

We have the following.
\begin{thm}
The one-point function $H_2(n;\ell)$ reads:
$$H_2(n;\ell)=\frac{\tilde U_{n,n}}{U_{n,n}}=\frac{1}{2^n}\sum_{k=0}^{\ell-n} {n\choose k}$$
\end{thm}
\begin{proof}
From the above definition of $\tilde A$, we see that $H_2(n;\ell)=0$ for $\ell<n$. Using the same $L$ as before,
we find that $L^{-1}{\tilde A}={\tilde U}$, with 
\begin{eqnarray*}
{\tilde U}_{n,n}&=&\sum_{k=0}^n (L^{-1})_{n,k} {\tilde A}_{k,n} 
=\sum_{k=\lfloor \frac{3n-\ell}{2}\rfloor}^{n} (-1)^{n+k} {n\choose k} {\ell+2k-2n \choose n}=
\sum_{k=0}^{\lfloor \frac{\ell-n}{2}\rfloor} (-1)^{k}{n\choose k} {\ell-2k \choose n}\\
&=&\sum_{k=0}^{\lfloor \frac{\ell-n}{2}\rfloor} (1-x)^n\Big\vert_{x^k} \frac{1}{(1-y)^{n+1}}\Bigg\vert_{y^{\ell-2k-n}} \\
&=& \sum_{k=0}^{\lfloor \frac{\ell-n}{2}\rfloor} (1-x^2)^n\Big\vert_{x^{2k}} \frac{1}{(1-y)^{n+1}}\Bigg\vert_{y^{\ell-2k-n}}
=\frac{(1-x^2)^n}{(1-x)^{n+1}}\Bigg\vert_{x^{\ell-n}}=\frac{(1+x)^n}{1-x}\Bigg\vert_{x^{\ell-n}}
=\sum_{k=0}^{\ell-n} {n\choose k} 
\end{eqnarray*}
where we have performed a change of variables $x\to x^2$ in the first factor, and used Lemma \ref{prodlem}.
The theorem follows.
\end{proof}

Taking the scaling limit: $n$ large, $p=zn$, $\ell=\xi n$, $k=x n$, we find that
$$H_2(n;\xi n)\sim\int_{0}^{\xi-1} dx\ e^{nS_0(x,\xi)},\quad S_0(x,\xi)=-(x{\rm Log}(x )+(1-x){\rm Log}(1-x))$$
and
$$W_2(n,zn;\xi \ell)\sim e^{nS_1(z,\xi)},\quad S_1(z,\xi)=(2-\xi){\rm Log}(2-\xi)-(z-1){\rm Log}(z-1)-(3-\xi-z){\rm Log}(3-\xi-z)$$

As before, we first extremize over $x$, with a saddle point at $x=x^*=1/2$. Two cases must be considered
(1) $\xi-1>1/2$, i.e.  $\xi>3/2$, then $S_0^*={\rm Log}(2)$.
(2) $\xi<3/2$, then $x^*=\xi-1$ and $S_0^*(\xi)=-(\xi-1){\rm Log}(\xi-1)-(2-\xi){\rm Log}(2-\xi)$.

Next we extremize over $\xi$. 
(1) if $\xi>3/2$ we find the saddle point equation $(3-\xi-z)=(2-\xi)$ which has no solution for $z>1$.
(2) hence $\xi<3/2$, and we have the saddle-point equation: $\partial_\xi S_0^*(\xi)+ \partial_{\xi} S_{1}(z,\xi)=0$,
namely:
$$ 3-\xi-z=\xi-1\ \  \Leftrightarrow \ \  \xi=2-\frac{z}{2}$$

The line through $(\xi,1)$ and $(2,z)$ has the equation:
$$ y- \frac{z-1}{2-\xi} x-\frac{2-z\xi}{2-\xi} = y- 2\frac{z-1}{z} x-2\frac{2-2 z+\frac{z^2}{2}}{z}=0$$
and the envelope is obtained by eliminating $z$ from this equation and its derivative w.r.t. $z$,
namely: 
$$ x=2-\frac{z^2}{2}\quad {\rm and}\quad y=z(2-z) \qquad (z\in [1,2])$$
Note that the limiting case $z=2$ corresponds to the tangent at  the origin, with slope $1$. 
As expected we recover the arctic parabola equation \eqref{parabo}:
$$-8 x + 4 x^2 + 8 y - 4 x y + y^2=0\ ,$$
now also valid for $0\leq x\leq 3/2$.

To summarize the methods employed above, we have represented in Fig.~\ref{fig:everyothercont}
the continuum limit of the arctic parabola
and two sample tangents obtained by the extremization procedures above.

\section{An interacting path model: Vertically Symmetric Alternating Sign Matrices}\label{sec:VSASM}

In this  last section we address the question of {\it interacting} lattice paths with the example of 
Vertically Symmetric Alternating Sign Matrices (VSASM). It turns out that VSASM, just like ordinary ASM
are in bijection with certain non-intersecting lattice paths called osculating paths, which may have "kissing points"
where two paths share a vertex without sharing edges, and without crossing. Assuming the applicability of the tangent method,
we will derive the arctic curve for large VSASM.

\label{sec:VSASM}

\subsection{Vertically symmetric alternating sign matrices}

An \emph{alternating sign matrix} (ASM) of size $n$ is an $n\times n$ matrix satisfying the following conditions:
\begin{itemize}
\item All matrix elements are equal to $1,0$, or $-1$.
\item The non-zero entries alternate in sign along each row and column.
\item The sum of all matrix elements in any row or column is equal to $1$.
\end{itemize}
Note that these three conditions taken together imply that in the first and last row of any ASM, and also in the first and 
last column, all matrix elements are equal to $0$ except for a single $1$.

A \emph{vertically symmetric alternating sign matrix} (VSASM) is an alternating sign matrix of size $2n+1$
which is also invariant under reflection about its middle column. We note that the properties of an ASM combined with the
vertically symmetric property of VSASMs implies that the central column of any VSASM has entries which alternate as
$1,-1,1,\dots,-1,1$. This implies that the unique $1$ in the top and bottom row of any VSASM is located in the middle entry.

We now collect several theorems on the enumeration and refined enumeration of ASMs and VSASMs which we make use of
in the Tangent method calculation for VSASMs later in this section. First we state the theorems on the number 
of ASMs of size $n$ and VSASMs of size $2n+1$ for any $n\in\mathbb{N}$. For ASMs we have:
\begin{thm}{\cite{zeilberger1,kuperberg1996another}}
Let $N_{\text{ASM}}(n)$ be the number of ASMs of size $n$. Then we have
\beq
	N_{\text{ASM}}(n)= \prod_{i=0}^{n-1}\frac{(3i+1)!}{(n+i)!}\ .
\eeq
\end{thm}
In the VSASM case we have: 
\begin{thm}{\cite{kuperberg2002symmetry,RS}}
Let $N_{\text{VSASM}}(2n+1)$ be the number of VSASMs of size $2n+1$. Then we have
\beq
	N_{\text{VSASM}}(2n+1)= \frac{1}{2^n}\prod_{i=1}^n\frac{(6i-2)!(2i-1)!}{(4i-1)!(4i-2)!}\ .
\eeq
\end{thm}

Next we consider the refined enumerations of ASMs and VSASMs. For ASMs recall that the top row has all entries equal to 
$0$ except for a single entry which is equal to $1$. We may therefore consider a refined enumeration in which we
compute the number of ASMs of size $n$ which have the unique $1$ in their top row in the $\ell^{th}$ position, for 
$\ell\in[1,n]$. Then we have the following theorem.
\begin{thm}{\cite{zeilberger2}}
Let $N_{\text{ASM}}(n,\ell)$ be the number of ASMs of size $n$ with the unique $1$ in the first row located in column 
$\ell$. Then we have
\beq
	N_{\text{ASM}}(n,\ell)= \frac{\binom{n+\ell-2}{n-1}\binom{2n-1-\ell}{n-1}}{\binom{3n-2}{n-1}}N_{\text{ASM}}(n)\ .
\eeq
\end{thm}
Note also that by the $90^{\circ}$ rotation symmetry of the ASM enumeration problem, $N_{\text{ASM}}(n,\ell)$
is also equal to the number of ASMs of size with their unique $1$ at position $\ell$ in the first column, or in the
last column, or in the last row. 

In the case of VSASMs we already remarked that the unique $1$ in the first and last rows must appear in the central column.
However, the unique $1$ in the first \emph{column} can still take any position (and likewise for the last column). 
We can therefore
consider a refined enumeration in which we compute the number of VSASMs of size $2n+1$ in which the unique
$1$ in the first column is at position $\ell$ for $\ell\in[1,2n+1]$. Then we have:
\begin{thm}{\cite{RS}}
Let $N_{\text{VSASM}}(2n+1,\ell)$ be the number of VSASMs of size $2n+1$ with the unique $1$ in the first column located in 
row $\ell$. Then we have
\beq
	N_{\text{VSASM}}(2n+1,\ell)= \frac{N_{\text{VSASM}}(2n-1)}{(4n-2)!}\sum_{i=1}^{\ell-1}(-1)^{\ell+i-1}\frac{(2n+i-2)!(4n-i-1)!}{(i-1)!(2n-i)!}\ .
\eeq
\end{thm}
Due the reflection symmetry of VSASMs, $N_{\text{VSASM}}(2n+1,\ell)$ is also the number of VSASMs of size $2n+1$ with the 
unique $1$ in the \emph{last} column located in row $\ell$. As mentioned before, such an alternating  formula is not suitable
for large size estimates, as terms may cancel out.
Fortunately, we will be able to use the following formula relating the refined enumerations of ASMs and VSASMs. More precisely, the
formula relates certain generating functions formed from the refined enumerations of ASMs and VSASMs. We have:
\begin{thm}{\cite{RS}}
\label{thm:raz-strog}
The following equality holds:
\beq
	\frac{1}{N_{\text{VSASM}}(2n-1)}\sum_{\ell=1}^{2n}N_{\text{VSASM}}(2n+1,\ell) t^{\ell-1}= \frac{1}{N_{\text{ASM}}(2n-1)}\frac{t}{t+1}\sum_{\ell=1}^{2n} N_{\text{ASM}}(2n,\ell)t^{\ell-1}\ .
\eeq
\end{thm}
Note also that on the left-hand side of this equation we can extend the range of the sum to $\ell\in[1,2n+1]$ due the
fact that $N_{\text{VSASM}}(2n+1,1)= N_{\text{VSASM}}(2n+1,2n+1)=0$ by the properties of VSASMs. 

\subsection{Mapping to six vertex model configurations}

The ASMs of size $n$ are in bijection with configurations of the six vertex model with \emph{domain wall} boundary conditions 
(DWBCs) on an $n\times n$ portion of the square lattice $\mathbb{Z}^2$. We now give a short review of the six vertex model
and the DWBCs, and then discuss the mapping from configurations of the six vertex model with DWBCs to ASMs.

The six vertex model is a two-dimensional statistical mechanical model in which the degrees of freedom are
arrows placed on the edges of the square lattice $\mathbb{Z}^2$. 
In addition, the arrow configurations are constrained to satisfy the \emph{ice rule}, which states that at each vertex the
number of incoming arrows must be equal to the number of outgoing arrows. On the square lattice there are exactly six
possible arrow configurations around a vertex that satisfy the ice rule, thus explaining the name of this model. The six possible
vertex configurations are traditionally divided into three types called $a$, $b$, and $c$-type vertices. Finally, there is also a 
representation of configurations of the six vertex model in terms of lattice paths. In this representation each possible 
configuration of arrows around a vertex is mapped to a segment of one or two lattice paths. We display the six
possible vertex configurations of this model, along with their type ($a$, $b$, or $c$) and their representation in terms of lattice
paths in Fig.~\ref{fig:6V-arrows}. The additional data in the last row of Fig.~\ref{fig:6V-arrows} (namely, the numbers 
$0,1$ or $-1$) is related to the mapping to ASMs and will be explained shortly.

\begin{figure}
\centering
\includegraphics[width=12.cm]{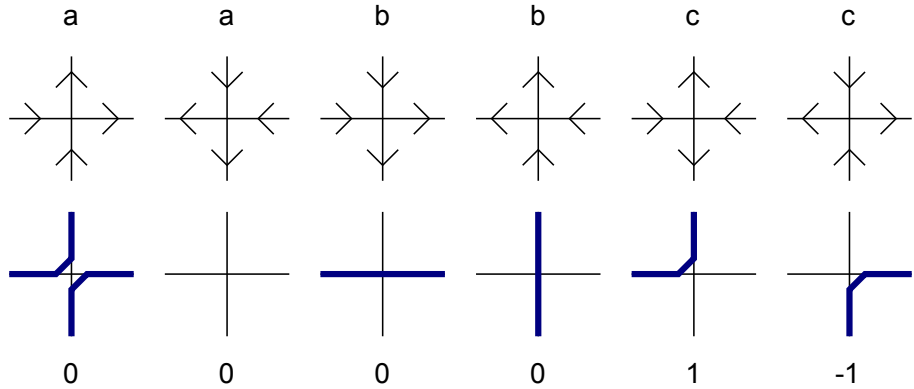}
\caption{\small The six kinds of vertex in the six vertex model, together with their corresponding lattice path configuration and
the matrix element (i.e., $1,0$, or $-1$) that they map to in the mapping from six vertex model configurations to ASMs.}
\label{fig:6V-arrows}
\end{figure}

We note here that in the lattice path formulation the six vertex model should be thought of as a model of \emph{interacting} 
lattice paths due to the fact that two paths are allowed to touch, but not cross, at a single vertex (see the first $a$-type 
vertex in  Fig.~\ref{fig:6V-arrows}). This should be contrasted with the other lattice path models studied in this
paper in which the paths are all non-intersecting. In non-intersecting lattice path models two paths are never allowed to meet 
at a single vertex. We may therefore characterize the non-intersecting lattice path models as models of \emph{non-interacting}
lattice paths.

We now describe the DWBCs for the six vertex model, which were first considered by Korepin~\cite{Korepin}. 
The six vertex model with DWBCs simply consists of the six vertex model formulated on an $n\times n$ region of the 
square lattice, in which the arrows on the edges which ``stick out" from the boundary are fixed to point \emph{in} towards
the domain on the left and right boundaries, and \emph{out} of the domain on the top and bottom boundaries. In this lattice
path formulation the DWBCs imply that on an $n\times n$ domain there are $n$ lattice paths which enter at the top of the
domain and exit on the left side of the domain. A sample configuration of the six vertex model with DWBCs, in both the
arrow and lattice paths formulations, is shown for a system of size $n=5$ in Fig.~\ref{fig:6V-config}.

To define the partition function for the six vertex model we assign weights to each of the six possible vertex configurations.
We will make a slight abuse of notation and use the letters $a$, $b$, and $c$ to also denote the weights of the vertices of
types $a$, $b$, and $c$, respectively. The weight $w(\mathcal{C})$ of a configuration $\mathcal{C}$ of the six vertex model
is defined to be equal to the product of the weights of all vertices in the configuration $\mathcal{C}$. The partition function
$Z_{\text{6VDMBC}}(n;a,b,c)$ for the six vertex model with DWBCs for a system of size $n$ is defined as
\beq
	Z_{\text{6VDMBC}}(n;a,b,c)= \sum_{\mathcal{C}}w(\mathcal{C})\ ,
\eeq
where the sum is taken over all configurations $\mathcal{C}$ consistent with the DWBCs. Here we have indicated explicitly that
the partition function $Z_{\text{6VDMBC}}(n;a,b,c)$ is a function of the weights $a$, $b$, and $c$. 

\begin{figure}
\centering
\includegraphics[width=12.cm]{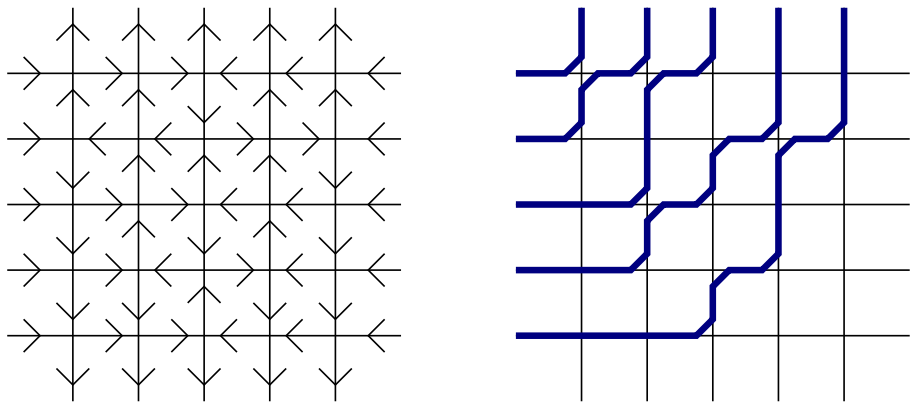}
\caption{\small The six vertex model configuration, in both the arrow and lattice path formulations, which is mapped to the
VSASM of size $5$ in Eq.~\eqref{eq:VSASM-size-5}.}
\label{fig:6V-config}
\end{figure}

It is known that the ASMs of size $n$ are in bijection with the configurations of the six vertex model on an $n\times n$ grid
and with DWBCs. 
The mapping from configurations of the six vertex model to ASMs is a set of rules that convert a vertex in the
six vertex model to a matrix element ($1,0$, or $-1$) in an ASM. We exhibit this rule in the last row of 
Fig.~\ref{fig:6V-arrows}. As an example, the six vertex model configuration shown in Fig.~\ref{fig:6V-config} maps to the
following ASM of size $5$:
\beq
	\begin{pmatrix}
	0 & 0&1&0&0 \\
	1&0&-1&0&1\\
	0&0&1&0&0\\
	0&1&-1&1&0 \\
	0&0&1&0&0
\end{pmatrix} \label{eq:VSASM-size-5}\ .
\eeq
We note here that this particular ASM is also a VSASM.

Since the ASMs of size $n$ are in bijection with the configurations of the six vertex model with DWBCs on a domain of size $n$,
we have that
\beq
	Z_{\text{6VDMBC}}(n;1,1,1)= N_{\text{ASM}}(n)\ .
\eeq
The particular parameter values $a=b=c=1$ correspond to the so-called \emph{ice point}. In what follows we will need a modified
partition function for the six vertex model with DWBCs in which we only sum over configurations $\tilde{\mathcal{C}}$
which can be mapped to a VSASM. Clearly, this modified partition function only exists for a square domain of odd size.
We denote this modified partition function by $\tilde{Z}_{\text{6VDMBC}}(2n+1;a,b,c)$, and we have
\beq
	\tilde{Z}_{\text{6VDMBC}}(2n+1;a,b,c)= \sum_{\tilde{\mathcal{C}}}w(\tilde{\mathcal{C}})\ ,
\eeq
where again the sum is only over those configurations of the six vertex model with DWBCs on a $(2n+1)\times(2n+1)$
grid which map to a VSASM. From the definition of this modified partition function it is also clear that
\beq
	\tilde{Z}_{\text{6VDMBC}}(2n+1;1,1,1)= N_{\text{VSASM}}(2n+1)\ .
\eeq

As a side note, we remark here that the usual approach to compute $N_{\text{VSASM}}(2n+1)$ and the refined enumeration
$N_{\text{VSASM}}(2n+1,\ell)$ uses a mapping from VSASMs to configurations of the six vertex model with ``U-turn" 
boundary conditions, as originally considered by Kuperberg~\cite{kuperberg2002symmetry}. 
We do not follow this approach here since the
only ingredient we will need for the application to the Tangent method is Theorem~\ref{thm:raz-strog} relating the refined
enumeration of VSASMs to that of ASMs. We now move on to the discussion of the Tangent method for this system.

\subsection{Tangent method}

We now apply the Tangent method to compute the arctic curve for VSASMs. We will prove that the arctic curve for VSASMs is
actually identical to the arctic curve of ordinary ASMs. Any possible differences between these two curves disappear in 
the scaling limit. We should mention that in the context of VSASMs, the arctic curve can be 
understood as separating a ``disordered" interior region of a VSASM, in which $1,0,$ and $-1$ entries occur, and an ordered
or frozen region in which only zeros appear. For VSASMs (and also for ASMs) the arctic curve will actually touch each of the four 
boundaries of the domain at a single point, and this point corresponds to the most likely location for the unique $1$ 
entry in the last row or column of the VSASM (or ASM). 

\begin{figure}
\centering
\includegraphics[width=12.cm]{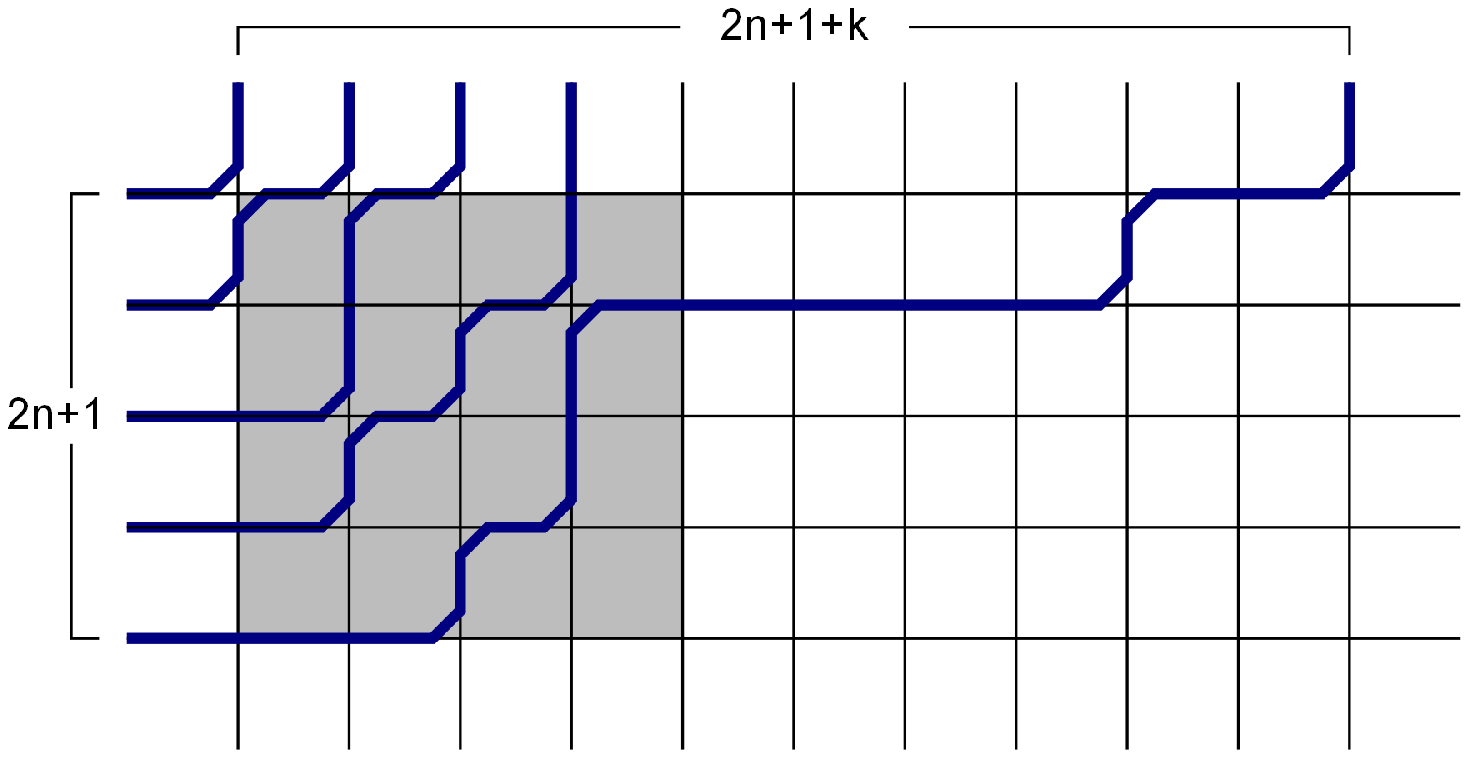}
\caption{\small The extended domain of size $(2n+1+k)\times n$ used for the Tangent method calculation for VSASMs (with
$2n+1=5$ and $k=6$).}
\label{fig:6V-extended}
\end{figure}

To apply the Tangent method we need to extend the original domain of the six vertex model with DWBCs. In their paper,
Colomo and Sportiello chose to extend the domain vertically. Here we choose to extend the domain horizontally instead, as
this type of extension is more natural for VSASMs. The reason is that in the case of VSASMs, the path (in the six 
vertex model formulation) which starts at the bottom left of the domain is constrained to turn upwards when it reaches
the center of the domain, and not before or after. Therefore it would not make sense to extend the domain vertically
by allowing this last path to dip below the original domain, since for VSASMs it would always be constrained to pass from the
original domain into the extended portion at the central column. On the other hand, this last path can arrive at the far right 
column of the original domain at any height $\ell>1$, so it does make sense to extend the domain horizontally and allow this last 
path to pass into the extended portion of the domain after hitting the far right column of the original domain at the height 
$\ell$. This is the approach that we take in this section.

To extend the domain we introduce a new integer $k>0$ and extend the domain to a region of the square lattice of 
size $(2n+1+k)\times (2n+1)$. We label the vertices in the domain by Cartesian coordinates $(x,y)=(i,j)$ with
$i=1,\dots,2n+1+k$ and $j= 1,\dots,2n+1$, and we consider $2n+1$ lattice paths on this domain. All $2n+1$ lattice 
paths enter at the left side of the domain. The top-most $2n$ lattice paths exit at the top of the domain at $x$ positions
$i=1,\dots,2n$, while the final lattice path exits the domain at the top of the very last column with coordinate $i=2n+1+k$.
In addition, we require that the portion of the lattice path configurations which lies is the original domain is such that  it
would map to a VSASM under the mapping from lattice paths/six vertex model configurations to ASMs (more precisely, it
should map to a VSASM if we turn the last path upwards at $i=2n+1$ so that it exits in the last column of the original 
domain). An example of an allowed lattice path configuration on this extended domain is shown in Fig.~\ref{fig:6V-extended}. 
It corresponds to the example shown in Fig.~\ref{fig:6V-config}, but where we let the last path cross into the extended part of 
the domain.

After rescaling all coordinates by a factor of $2n+1$, the original six vertex model domain is rescaled to lie in the unit square
$x\in[0,1]$, $y\in[0,1]$. Due the way we have chosen to extend the domain, the Tangent method will naturally 
produce the portion of the arctic curve for VSASMs which lies in the quadrant $x\in[\frac{1}{2},1]$, $y\in[0,\frac{1}{2}]$,
and we will have to keep this in mind when comparing our final expression for the arctic curve with other expressions in the
literature for the arctic curve of ASMs. Since any given VSASM has a ``mirror image" obtained by reflecting this VSASM
about the central \emph{row}, we can argue by symmetry that the portion of the arctic curve in the region
$x\in[\frac{1}{2},1]$, $y\in[\frac{1}{2},1]$ can be obtained from the curve in the region 
$x\in[\frac{1}{2},1]$, $y\in[0,\frac{1}{2}]$ by simply reflecting it over the line $y=\frac{1}{2}$. Finally, due to the defining
symmetry of VSASMs (which is the symmetry of reflection about the central \emph{column}), the entire arctic 
curve in the region with $x\leq \frac{1}{2}$ is obtained from the arctic curve in the region $x\geq \frac{1}{2}$ by reflection
over the line $x=\frac{1}{2}$.

As in previous sections, we compute the full partition function for the problem on the extended domain by expanding in the
vertical position $\ell$ at which the last path crosses from the original domain into the extended domain. In what follows
we use the notation $\tilde{n}:= 2n+1$ for convenience. If we denote by $N(\tilde{n},k)$ the partition function for the
six vertex model on the extended domain, then we can write
\beq
	N(\tilde{n},k)= \sum_{\ell=1}^{\tilde{n}}Z(\tilde{n},\ell)Y(\ell,k)\ .
\eeq
Here $Z(\tilde{n},\ell)$ is the partition function for paths on the $\tilde{n}\times\tilde{n}$ portion of the domain in which
the last path exits this portion of the domain at the height $\ell$. The other factor $Y(\ell,k)$ is the partition function
for the piece of the last path which lies outside of the original $\tilde{n}\times\tilde{n}$ portion of the domain. 
For the application to VSASMs we only need to consider the six vertex model at the ice point $a=b=c=1$, and so we restrict
our attention to this case in what follows.

Since we work at the ice point, we find that $Z(\tilde{n},\ell)$ is exactly equal to the number of VSASMs of size $\tilde{n}$
whose unique $1$ in the last column is at position $\ell$, 
\beq
	Z(\tilde{n},\ell)= N_{\text{VSASM}}(\tilde{n},\ell)\ .
\eeq
Next, $Y(\ell,k)$ is equal to the total number of six vertex model paths which start at $(\tilde{n},\ell)$, 
end at $(\tilde{n}+k,\tilde{n}+1)$, and are constrained to have the first step be horizontal and the last step be vertical. This
constraint then implies that $Y(\ell,k)$ is just equal to the total number of six vertex model paths
which start at $(\tilde{n}+1,\ell)$ and end at $(\tilde{n}+k,\tilde{n})$. By a ``six vertex model path" we mean a path
whose segments are constrained to be of the type illustrated in Fig.~\ref{fig:6V-arrows}. In addition, if we have only a single 
path, then we are restricted to the vertices of type $b$ and $c$ from Fig.~\ref{fig:6V-arrows}.
This calculation is presented in, for example, Appendix 1 of Ref.~\cite{COSPO}, and the answer can be expressed 
as
\beq
	Y(\ell,k)= \sum_{p=0}^{\text{min}(k-1,\tilde{n}-\ell)}\binom{k-1}{p}\binom{\tilde{n}-\ell}{p}\ .
\eeq
Here the summation variable $p$ can be interpreted as counting the number of ``north-east corners" of a path, where
we have a $(0,1)$ step immediately followed by a $(1,0)$ step (alternatively, the second type $c$ vertex from 
Fig.~\ref{fig:6V-arrows}). Thus, we find that the total partition function of the six vertex model (with our
restriction to configurations in the $\tilde{n}\times\tilde{n}$ portion such that they map to a VSASM) at the ice point on the 
extended domain is
\beq
	N(\tilde{n},k)= \sum_{\ell=1}^{\tilde{n}} \sum_{p=0}^{\text{min}(k-1,\tilde{n}-\ell)} N_{\text{VSASM}}(\tilde{n},\ell)\binom{k-1}{p}\binom{\tilde{n}-\ell}{p}\ .
\eeq
For later use we also define the one-point function $H(\tilde{n},\ell)$ as
\beq
	H(\tilde{n},\ell)= \frac{N_{\text{VSASM}}(\tilde{n},\ell)}{N_{\text{VSASM}}(\tilde{n})}\ .
\eeq

To compute the asymptotics of this partition function we first introduce new scaled variables $z,\xi$, and $\eta$ as
$k=\tilde{n}z$, $\ell=\tilde{n}\xi$, and $p= \tilde{n}\eta$. We then follow Colomo and Sportiello and define the
``free energy"
\beq
	F[z]= \lim_{\tilde{n}\to\infty}\frac{1}{\tilde{n}}{\rm Log}\left[\frac{N(\tilde{n},k)}{N_{\text{VSASM}}(\tilde{n})}\right]
\eeq
and the ``action"
\beq
	S(\xi,\eta;z)= \lim_{\tilde{n}\to\infty}\frac{1}{\tilde{n}}{\rm Log}\left[\binom{k-1}{p}\binom{\tilde{n}-\ell}{p} H(\tilde{n},\ell) \right]\ .
\eeq
After rescaling of the variables and applying Stirling's formula ${\rm Log}(q!)\approx q{\rm Log}(q)-q$ for large $q$, the action takes the form
\beq
	S(\xi,\eta;z)= f(z)-2f(\eta)-f(z-\eta) + f(1-\xi)-f(1-\xi-\eta) + \lim_{\tilde{n}\to\infty}\frac{1}{\tilde{n}}{\rm Log}\left[H(\tilde{n},\xi\tilde{n})\right]\ ,
\eeq
where we used the abbreviated notation $f(x):= x{\rm Log}(x)-x$.
In the scaling limit the free energy is dominated by a contribution from the stationary point of the action, which can
be found by solving the equations
\begin{align}
	\frac{\pd S(\xi,\eta;z)}{\pd \xi} &= 0 \\
	\frac{\pd S(\xi,\eta;z)}{\pd \eta} &= 0\ ,
\end{align}
for the saddle point solutions $\xi_{sp}(z)$ and $\eta_{sp}(z)$ which are functions of $z$.

Before rescaling of the coordinates, the tangent line we are looking for starts at the point $(\tilde{n},\ell)$ and ends at
$(\tilde{n}+k,\tilde{n})$. After rescaling and solving for $\xi_{sp}(z)$, the most likely value of the rescaled $y$
coordinate $\ell$ where the last path leaves the original domain, this line is defined by the equation
\beq
	y= \left(\frac{1-\xi_{sp}(z)}{z}\right)x +\xi_{sp}(z) - \left(\frac{1-\xi_{sp}(z)}{z}\right)\ .
\eeq
We now proceed with the solution of the saddle point equations.

After differentiating with respect to $\xi$ and $\eta$, the two saddle point equations take the form
\beq
	{\rm Log}\left(\frac{1-\xi-\eta}{1-\xi}\right) + \lim_{\tilde{n}\to\infty}\frac{1}{\tilde{n}}\frac{\pd}{\pd\xi}{\rm Log}\left[H(\tilde{n},\xi\tilde{n})\right]= 0
\eeq
and
\beq
	{\rm Log}(z-\eta)-2{\rm Log}(\eta)+{\rm Log}(1-\xi-\eta)= 0\ .
\eeq
The second equation can be solved immediately and allows us to express $\eta_{sp}(z)$ in terms of $\xi_{sp}(z)$ as
\beq
	\eta_{sp}(z)= \left(\frac{1-\xi_{sp}(z)}{1-\xi_{sp}(z)+z}\right)z\ .
\eeq
Solving the first saddle point equation is more complicated due to the presence of the term involving the one point function
$H(\tilde{n},\xi\tilde{n})$. Before presenting the solution we first eliminate $\eta$ from this equation using the solution 
of the second saddle point equation to find
\beq\label{sapot}
	{\rm Log}\left(\frac{1-\xi}{1-\xi+z}\right) + \lim_{\tilde{n}\to\infty}\frac{1}{\tilde{n}}\frac{\pd}{\pd\xi}{\rm Log}\left[H(\tilde{n},\xi\tilde{n})\right]= 0\ .
\eeq
To solve this equation we follow the method of Ref.~\cite{COSPO} and find
that the solution $\xi_{sp}(z)$ is given implicitly by the equation
\beq\label{csone}
	\xi_{sp}(z)= r(t)\ ,
\eeq
obtained by inverting the relation:
\beq\label{cstwo}
	t=\frac{1-\xi_{sp}(z)}{1-\xi_{sp}(z)+z}\ ,
\eeq
with
\beq
	r(t)= \lim_{\tilde{n}\to\infty}\frac{1}{\tilde{n}} t\frac{d}{dt}{\rm Log}(h_{\tilde{n}}(t))\ ,
\eeq
and where we introduced the generating function
\beq
	h_{\tilde{n}}(t):= \sum_{\ell=1}^{\tilde{n}}H(\tilde{n},\ell)t^{\ell-1}
\eeq
for the one point function $H(\tilde{n},\ell)$. To obtain eqns.(\ref{csone}-\ref{cstwo}), we simply estimate for large $\tilde n$:
$$
h_{\tilde{n}}(t)\sim   \int_0^1 d\xi\ H(\tilde{n},\xi\tilde{n})e^{{\tilde n}\xi {\rm Log}(t)}=\int_0^1 d\xi\ e^{{\tilde n}S(\xi,t)}
$$
with 
$$S(\xi,t)=\xi {\rm Log}(t) +\frac{1}{\tilde{n}}{\rm Log}\left[H(\tilde{n},\xi\tilde{n})\right] \ .
$$
If we compare the saddle-point equation $\partial_\xi S(\xi,t)=0$ with Eq.~\eqref{sapot}, then we find a relation between 
$\xi_{sp}(z)$ and $t$, namely \eqref{cstwo}.
Finally, the large $\tilde n$ estimate of $h_{\tilde{n}}(t)$ is given by:
$$h_{\tilde{n}}(t)\sim e^{{\tilde n} S(\xi_{sp}(z),t)}$$
therefore 
\begin{eqnarray*}
&&r(t):=\lim_{\tilde{n}\to\infty}\frac{1}{\tilde{n}} t\frac{d}{dt}{\rm Log}(h_{\tilde{n}}(t))=\lim_{\tilde{n}\to\infty} \left(\xi+t\frac{d\xi}{dt}{\rm Log}(t) +\frac{1}{\tilde{n}} t\frac{d}{dt}{\rm Log}\left[H(\tilde{n},\xi\tilde{n})\right]\right)\Big\vert_{\xi=\xi_{sp}(z)}\\
&&\qquad = \xi_{sp}(z)+t\frac{d\xi_{sp}(z)}{dt}\left({\rm Log}(t)+\lim_{\tilde{n}\to\infty}\frac{1}{\tilde{n}} \partial_{\xi} {\rm Log}\left[H(\tilde{n},\xi\tilde{n})\right]\right)\Big\vert_{\xi=\xi_{sp}(z)}=\xi_{sp}(z)
\end{eqnarray*}
by use of the saddle point equation \eqref{sapot} and the relation \eqref{cstwo}. This gives \eqref{csone}.

We now note that due to the relation
\beq
	\frac{1-\xi_{sp}(z)}{z}= \frac{t}{1-t}\ ,
\eeq
and the fact that $\xi_{sp}(z)=r(t)$, we can eliminate $\xi_{sp}(z)$ and $z$ in favor of $t$ and $r(t)$ in our expression for the 
tangent line. After this change of variables the equation for the tangent line takes the form
\beq
	y= \left(\frac{t}{1-t}\right)x + r(t) -\frac{t}{1-t}\ .
\eeq
To solve for the arctic curve we define the function 
\beq
	F(x,y;t)= - y + \left(\frac{t}{1-t}\right)x + r(t) -\frac{t}{1-t}
\eeq
such that the equation $F(x,y;t)= 0$ is the equation defining the tangent line. To extract the arctic curve
we then solve the simultaneous equations $F(x,y;t)= 0$ and $\frac{\pd}{\pd t}F(x,y;t)= 0$ to solve for 
$x$ and $y$ in terms of $t$, which then yields a parametric form $(x(t),y(t))$ for the arctic curve.

Our final task is to solve for the function $r(t)$. To start, we write out the generating function $h_{\tilde{n}}(t)$
in more detail (using the definition of the one point function) as
\beq
	h_{\tilde{n}}(t)= \sum_{\ell=1}^{\tilde{n}}\frac{N_{\text{VSASM}}(\tilde{n},\ell)}{N_{\text{VSASM}}(\tilde{n})}t^{\ell-1}\ .
\eeq
Next we apply Theorem~\ref{thm:raz-strog} to find that (recall that $\tilde{n}=2n+1$)
\beqa
	h_{\tilde{n}}(t) &=& \frac{N_{\text{VSASM}}(\tilde{n}-2)}{N_{\text{VSASM}}(\tilde{n})N_{\text{ASM}}(\tilde{n}-2)}\frac{t}{t+1}\sum_{\ell=1}^{\tilde{n}-1}N_{\text{ASM}}(\tilde{n}-1,\ell)t^{\ell-1}\nnb \\
	&=& \frac{N_{\text{VSASM}}(\tilde{n}-2)N_{\text{ASM}}(\tilde{n}-1)}{N_{\text{VSASM}}(\tilde{n})N_{\text{ASM}}(\tilde{n}-2)}\frac{t}{t+1}\sum_{\ell=1}^{\tilde{n}-1}\frac{N_{\text{ASM}}(\tilde{n}-1,\ell)}{N_{\text{ASM}}(\tilde{n}-1)}t^{\ell-1} \nnb \\
	&:=& \frac{N_{\text{VSASM}}(\tilde{n}-2)N_{\text{ASM}}(\tilde{n}-1)}{N_{\text{VSASM}}(\tilde{n})N_{\text{ASM}}(\tilde{n}-2)}\frac{t}{t+1} h^{\text{ASM}}_{\tilde{n}-1}(t)\ ,
\eeqa
where we defined
\beq
	h^{\text{ASM}}_{\tilde{n}-1}(t)= \sum_{\ell=1}^{\tilde{n}-1}\frac{N_{\text{ASM}}(\tilde{n}-1,\ell)}{N_{\text{ASM}}(\tilde{n}-1)}t^{\ell-1}\ .
\eeq

We now compute the function $r(t)$ as
\beqa
	r(t) &=& \lim_{\tilde{n}\to\infty}\frac{1}{\tilde{n}} t\frac{d}{dt}{\rm Log}(h_{\tilde{n}}(t)) \nnb \\
	&=& \lim_{\tilde{n}\to\infty}\frac{1}{\tilde{n}}\left[ \frac{1}{1+t} + t\frac{d}{dt}{\rm Log}(h^{\text{ASM}}_{\tilde{n}-1}(t))\right]\ .
\eeqa
The term $\frac{1}{1+t}$ came from differentiating the factor of $\frac{t}{t+1}$ which appeared in 
Theorem~\ref{thm:raz-strog}. In the $\tilde{n}\to\infty$ limit this term goes to zero due to the prefactor
of $\frac{1}{\tilde{n}}$ and, since $\tilde{n}-1$ can be replaced with $\tilde{n}$ in this limit, we find that
\beq
	r(t)= r^{\text{ASM}}(t)\ ,
\eeq
where
\beq
	 r^{\text{ASM}}(t) = \lim_{\tilde{n}\to\infty}\frac{1}{\tilde{n}} t\frac{d}{dt}{\rm Log}(h^{\text{ASM}}_{\tilde{n}}(t))\ .
\eeq
In the ASM case an explicit expression for this function is known (see, for example, Ref.~\cite{CP2010}) and we have
\beq
	 r^{\text{ASM}}(t) =  \frac{\sqrt{t^2-t+1}-1}{t-1}\ .
\eeq

To close this section we present the explicit form of the arctic curve for VSASMs which, due to the relation 
$r(t)= r^{\text{ASM}}(t)$, is identical to the arctic curve for ordinary ASMs. Solving the simultaneous equations
$F(x,y;t)= 0$ and $\frac{\pd}{\pd t}F(x,y;t)= 0$ using the explicit form of $r^{\text{ASM}}(t)$ yields the
parametric form of the arctic curve,
\beqa
	x(t) &=& \frac{1+t}{2\sqrt{1-t+t^2}} \\
	y(t) &=& \frac{-2+t+2\sqrt{1-t+t^2}}{2\sqrt{1-t+t^2}}\ .
\eeqa
It is now straightforward to verify that $x(t)$ and $y(t)$ satisfy the equation 
\beq
	4(1-x) - 4(1-x)^2 + 4y - 4y^2 + 4(1-x)y -1 = 0\ ,
\eeq
which is exactly the equation for the portion of the arctic curve for ASMs in the region $x\in[\frac{1}{2},1]$, 
$y\in[0,\frac{1}{2}]$ of the rescaled domain. Note that in Ref.~\cite{CP2010} Colomo and Pronko obtained the portion
of the arctic curve for ASMs in the region $x\in[0,\frac{1}{2}]$, $y\in[0,\frac{1}{2}]$, and so our curve differs from theirs
by the replacement $x\to1-x$, which implements the reflection over the line $x=\frac{1}{2}$. It is also interesting and
worth noting that we still obtain the correct arctic curve for ASMs and VSASMs by extending the domain of the six vertex
model horizontally instead of vertically as in Ref.~\cite{COSPO}. Finally, let us also note that due to the
term $\frac{t}{1+t}$ appearing in Theorem~\ref{thm:raz-strog}, one would naively think that the arctic curves for ASMs and
VSASMs would be different. However, the contribution of this extra term to the function $r(t)$ is suppressed in the
$\tilde{n}\to\infty$ limit, and so the arctic curves for ASMs and VSASMs turn out to be identical. However, \emph{finite-size} 
numerical studies of the arctic curve in VSASMs should show slight differences from the arctic curve for ASMs, due to this
extra term.

\section{Conclusion}\label{conclusec}

\subsection{Summary}
In this paper, we have explored four specific statistical models that can be rephrased into (weakly) non-intersecting path models,
and shown how to apply the tangent method for determining the arctic curve in the limit of large size. Our exact results
confirm the applicability of the method in the case of the domino tilings of the Aztec diamond (Section~\ref{aztecsec}), 
where the arctic circle was derived rigorously by other methods \cite{CEP,DFSG}. 
Our Dyck path model for the rhombus tiling of a half-hexagon (Sections~\ref{dycksecone} and \ref{dycksectwo})
and the osculating path model for VSASM (Section~\ref{sec:VSASM})
both lead to analogous results: the arctic curve for a half-domain with fixed 
boundary conditions along the cut is identical to that of the full domain (ellipse for the hexagon, and 4 quarters of ellipse
for ASM). Finally, the other path model of Section \ref{everyothersec} leads to a parabola, which, despite the fact that the
partition function matches that of domino tilings of the Aztec diamond, points to a very different asymptotic behavior.

\subsection{Open problems}
We are left with many open questions. The tangent method itself remains to be proved rigorously.
Regarding its range of applicability, it seems to not only apply to strictly non-intersecting paths, but to weakly interacting
ones as well, for which kissing or osculating points are allowed. We may perhaps test this on other path models, such
as osculating large Schroeder paths, corresponding to tilings of the Aztec diamond or other domains by means of $2\times 1$
and $1\times 2$ dominos and a finite set of extra larger tiles that account for the various kissing point configurations. Note in that
case the novel possibility for three paths to share an osculating vertex. 
Other cases of interest are the path models corresponding to higher spin versions of the 6 Vertex model (such as the spin-1
19 vertex model for instance). Finally, another set of problems concerns path models with inhomogeneous weights, namely
with steps of different weights depending on their positions. For instance periodic inhomogeneous weight domino tilings of the Aztec
diamond were studied in \cite{DFSG}, and shown to give rise to more involved arctic curves, including ``bubbles" of intermediate
disorder semi-crystalline phases. We intend to return to these problems in a future publication.

\bibliographystyle{amsalpha} 
\bibliography{tangent-dyck-refs}

\end{document}